\documentclass[11pt]{article}
\usepackage{graphicx}
\usepackage{float}
\usepackage{subfigure}
\usepackage{caption}
\usepackage{multirow}
\usepackage{makecell}
\usepackage{appendix}
\usepackage[figuresright]{rotating}
\usepackage{booktabs}
\usepackage[linesnumbered,ruled]{algorithm2e}
\usepackage[margin=1in]{geometry}
\usepackage{hyperref}
\usepackage{amsfonts}
\usepackage{mathrsfs}
\usepackage{comment}
\usepackage{amsmath}
\usepackage{amssymb}
\usepackage{amsthm}
\usepackage{amscd}
\usepackage{graphicx}
\usepackage{indentfirst}
\usepackage[all]{xy}
\usepackage{titlesec}
\usepackage{enumerate}
\usepackage{bm}
\usepackage{enumitem}
\usepackage{color}
\usepackage{dsfont}
\usepackage{arydshln}
\usepackage{booktabs}
\newtheorem{theorem}{Theorem}[section]
\newtheorem{lemma}[theorem]{Lemma}
\newtheorem{proposition}[theorem]{Proposition}
\newtheorem{corollary}[theorem]{Corollary}

\newtheorem{definition}[theorem]{Definition}

\newtheorem{remark}[theorem]{Remark}

\newtheorem{example}[theorem]{Example}

\begin{document}
\title{ On Twisted Roth-Lempel Codes\footnote{The research was supported by the National Natural Science Foundation of China under the Grants 12222113 and 12441105.}}
\author{Huiyue Lei\footnote{Huiyue Lei is with the School of Mathematical Sciences, Capital Normal University, Beijing 100048, China. Email:362885063@qq.com },
   \and Haojie Gu\footnote{Haojie Gu is with the School of Mathematical Sciences, Capital Normal University, Beijing 100048, China. Email: 2200502051@cnu.edu.cn.},
	\and Jun Zhang\footnote{Jun Zhang is with the School of Mathematical Sciences, Capital Normal University, Beijing 100048, China. Email: junz@cnu.edu.cn.},
    \and Haiyan Zhou\footnote{Haiyan Zhou is  with the School of Mathematical Sciences, Nanjing Normal University, Nanjing 210023, China. Email:zhouhy@njnu.edu.cn.}
}

\date{}
\maketitle

\begin{abstract}
	In 1989, Roth and Lempel constructed a well-known family of non-Reed-Solomon maximum distance separable (MDS) codes. For decades, this family of codes has attracted extensive research attention due to its algebraic structure, low-complexity decoding, and broad applications in cryptography and data storage. In this paper, we present a class of twisted Roth-Lempel codes. We investigate their minimum distance, MDS  and NMDS properties. Specifically, we determine the necessary and sufficient conditions for the $\operatorname{TRL}$ codes to have minimum distance $n-k$ or $n-k+1$. Furthermore, we determine the necessary and sufficient conditions for the  $\operatorname{TRL}$ code to be an MDS or NMDS code. Moreover, we show that the dimension of the Schur square of the $\operatorname{TRL}$ code is at least $2k+1$, and thus the code $\operatorname{TRL}_{k,n+2}(\boldsymbol{\alpha},\eta,\delta,\ell)$ is a non-RS 
    code inequivalent to the corresponding RL code.

	\begin{flushleft}
	\end{flushleft}
\end{abstract}

\section{Introduction}

Let $\mathbb{F}_{q}$ be the  finite field with $q$ elements and $\mathbb{F}_{q}^{*}=\mathbb{F}_{q}\backslash\{0\}$, where $q$ is a power of the prime $p$. Let $\mathbb{F}_{q}^n$ be the $n$-dimensional vector space over the finite field $\mathbb{F}_{q}$. For any vector $ \boldsymbol{x}=(x_1,x_2,\cdots,x_n)\in \mathbb{F}_{q}^n$, the \emph{Hamming weight} $wt( \boldsymbol{x})$ of $ \boldsymbol{x}$ is defined to be the number of non-zero coordinates, i.e.,
$$wt( \boldsymbol{x})=|\left\{i\,|\,1\leqslant i\leqslant n,\,x_i\neq 0\right\}|.$$

An $[n,k,d]$-linear code $\mathcal{C}\subseteq\mathbb{F}_{q}^n$ is a $k$-dimensional linear subspace of $\mathbb{F}_{q}^n$ whose minimum distance $d=d(\mathcal{C})$ is given by
$$d(\mathcal{C})=\min\left\{wt(\boldsymbol{c}):\boldsymbol{c}\in\mathcal{C}\backslash\{0\}\right\}.$$
 The dual code of C is defined as
 \[
 \mathcal{C}^{\perp}=\left\{(x_1,\cdots,x_n)\in\mathbb{F}_{q}^n:\sum\limits_{i=1}^nx_iy_i=0,\ \mbox{for all}\ (y_1,\cdots,y_n)\in\mathcal{C}\right\}.
 \]
The well-known Singleton bound states that $d\leq n-k+1$ for any linear code $\mathcal{C}$ with parameters $[n,k,d]$. The non-negative integer $S(\mathcal{C})=n-k+1-d$ is called the Singleton defect of the code $\mathcal{C}$~\cite{de1996almost}. If $S(\mathcal{C})=0$, then $\mathcal{C}$ is called
a maximum distance separable (MDS) code. If $S(\mathcal{C})=1$, then $\mathcal{C}$ is called an
almost-MDS (AMDS) code. If $S(\mathcal{C})=S(\mathcal{C}^{\perp})=1$, then $\mathcal{C}$ is called a near-MDS (NMDS) code. More generally, if $S(\mathcal{C})=S(\mathcal{C}^{\perp})=m$, then $\mathcal{C}$ is called
$m$-MDS. Due to their optimal or near-optimal distance properties, MDS and NMDS codes have been extensively studied and have found important applications in communication systems, distributed storage, cryptography, combinatorial designs, and quantum error correction~\cite{huffman2010fundamentals,macwilliams1977theory,simos2012mds,thomas2018binary}. 

\begin{definition}
    Let $\boldsymbol{\alpha}=\{\alpha_{1},\cdots,\alpha_{n}\}\subseteq\mathbb{F}_{q}$ be the evaluation set and $\boldsymbol{v}=(v_{1},\cdots,v_{n})\in (\mathbb{F}_{q}^{*})^n$, where $\alpha_{i}\neq\alpha_{j}$ for all $i\neq j$. Then the generalized Reed-Solomon code $GRS_{k}(\boldsymbol{\alpha},\boldsymbol{v})$ of length $n$ and dimension $k$ is defined as 
    \begin{equation*}
        GRS_{k}(\boldsymbol{\alpha},\boldsymbol{v})=\left\{(v_{1}f(\alpha_{1}),\cdots,v_{n}f(\alpha_{n})):f(x)\in\mathbb{F}_{q}[x]_{<k}\right\},
    \end{equation*}
    where $\mathbb{F}_{q}[x]_{<k}:=\left\{f(x)\in\mathbb{F}_{q}[x]:\deg(f(x))<k \right\}$. If $\boldsymbol{v}$ is the all-one vector, it is called a Reed-Solomon code of length $n$ and dimension $k$, denoted as $RS_{k}(\boldsymbol{\alpha})$.
\end{definition}

The most prominent examples of MDS codes are Reed–Solomon (RS) codes and their generalizations. Their algebraic structure, optimal distance, and efficient decoding algorithms make them indispensable in both theory and practice. Nevertheless, RS-type codes are highly structured, and this rigidity has motivated the search for new MDS code families that are not equivalent to RS codes, which are called  non-RS MDS codes. The construction of non-RS MDS codes is not only a classical problem in coding theory and finite geometry~\cite{chen2023many,jin2025new,roth1989construction}, but also of practical relevance to code-based cryptography.

A milestone in this direction was the introduction of Roth–Lempel (RL) codes by Roth and Lempel in 1989~\cite{roth1989construction}. The classical RL code is constructed by augmenting a $[n, k]$ RS code with two additional coordinates, yielding a $[n+2, k]$ code with generator matrix of the following form

$$
G = \begin{pmatrix} G_{RS}(\bm{\alpha}) & \mathbf{0}_{(k-2)\times 2} \\ & T_{2\times 2}(\delta) \end{pmatrix}
$$
where $G_{RS}(\bm{\alpha})$ is a generator matrix of an RS code and $T_{2\times 2}(\delta) = \begin{pmatrix} 0 & 1 \\ 1 & \delta \end{pmatrix}$. The MDS property holds if and only if $-\delta \notin \Delta_k$, where $\Delta_k$ is a subset of the field determined by the evaluation points. The generalized Roth-Lempel (GRL) code~\cite{li2025new,liang2025equivalent} extends this construction by replacing $T_{2\times 2}(\delta)$ with an arbitrary invertible $\ell \times \ell$ matrix $\bm{A}_{\ell\times\ell}$, producing $[n+\ell, k]$ codes with significantly expanded parameter flexibility. When the MDS condition is not satisfied, GRL codes frequently exhibit near-MDS (NMDS) properties—being almost-MDS (AMDS) with AMDS duals. Liang et al.~\cite{liang2025equivalent}
have established explicit constructions of the NMDS GRL and extended GRL (EGRL) codes with completely determined weight distributions, enabling their application in combinatorial design theory.  Liu et al.~\cite{liu2026generalized} established the necessary and sufficient NMDS conditions for the GRL codes with $\ell=2$ and $\ell=3$.

A defining characteristic of the RL and GRL codes is their non-GRS nature for $k > \ell$ . This structural distinction is of paramount importance for cryptographic applications. GRS codes, despite their optimal parameters and efficient decoding, are vulnerable to the Sidelnikov-Shestakov structural attack, which recovers the secret key in polynomial time~\cite{sidelnikov1992insecurity,wieschebrink2006attack} . In contrast, RL and GRL codes resist such attacks due to their fundamentally different algebraic structure, making them valuable candidates for code-based cryptography, particularly in the context of post-quantum security.

More recently, twisted constructions have proven to be particularly effective in producing new algebraic codes with rich structures. In 2016, Sheekey \cite{Sheekey2016anew} introduced a new class of maximum rank distance codes, known as Twisted Gabidulin codes, which are MDS with respect to the rank metric and were shown to be inequivalent to Gabidulin codes (the rank metric analog of Reed-Solomon codes). Inspired by Sheekey's work, Beelen et al. \cite{beelen2017twisted,beelen2022twisted} proposed Twisted Reed-Solomon (TRS) codes and demonstrated that certain families of TRS codes are non-RS MDS codes. Following this development, extensive research has been devoted to the structure and properties of TRS codes \cite{cheng2023parity,ding2025new,gu2023twisted,huang2021mds,sui2022mds1,sui2023new,sui2022mds2,zhang2022class,zhu2024class}. 

The main contributions of this paper are as follows. 
First, we introduce a family of twisted Roth-Lempel codes and determine lower bounds for their minimum distance in all hook positions. Second, we give necessary and sufficient conditions under which these codes 
have minimum distance $n-k$ or $n-k+1$. 
Third, we derive explicit MDS criteria in terms of elementary symmetric 
functions of the evaluation subsets. 
Fourth, we compute Schur square dimensions, which provide distinguishers 
from generalized Reed-Solomon codes and from the corresponding classical 
Roth-Lempel codes. 
Finally, we give explicit field-extension and subgroup-based constructions 
of non-Reed-Solomon MDS codes and characterize the NMDS property for the 
boundary cases $\ell=k-2$ and $\ell=k-1$.


\section{Preliminaries}\label{sec2}
In this section, we establish the notation for the rest of the paper. By “natural numbers”, we mean  positive integers, i.e., $\mathbb{N}^{+}=\left\{1,2,3,\cdots\right\}$.  The set of non-negative integers is denoted by $\mathbb{N}$. For $m\in\mathbb{N}^{+}$, let $[m]$ denote the set of integers from $1$ to $m$ and let $[0,m]$ denote the set of integers from $0$ to $m$, i.e., $[m]:=\left\{1,2,\cdots,m\right\}$ and $[0,m]:=\left\{0,1,\cdots,m\right\}$. Given an  $n$-subset $U=\left\{\beta_{1},\beta_{2},\cdots,\beta_{n}\right\}\subseteq \mathbb{F}_{q}$ and $k$-subset $I=\left\{i_{1},\cdots,i_{k}\right\}\subseteq [n]$, let $U_{I}=\left\{\beta_{i_{1}},\beta_{i_{2}},\cdots,\beta_{i_{k}}\right\}$.
 In particular, for $j_{1}\neq j_{2} \in [n]$, let $U_{j_{1}}=U\setminus\{\beta_{{j_1}}\}$ and $U_{j_{1},j_{2}}=U\setminus\{\beta_{j_1},\beta_{j_2}\}$. For the $k$-subset $I=\{i_1,\cdots,i_{k}\}$, let
$V(U_{I}):=V(\beta_{i_{1}},\cdots,\beta_{i_{k}})=\prod\limits_{1\leq j_{1}<j_{2}\leq k}(\beta_{i_{j_{2}}}-\beta_{i_{j_{1}}}),S_{j}(U_{I})=\sum\limits_{\left|J\right|=j\atop J\subseteq [k]}\prod\limits_{s\in J}\beta_{i_{s}}$ and $\sigma_{j}(U_I)=(-1)^j S_{j}(U_I)$, where $S_0(U_I)=\sigma_{0}(U_I)=1$ and $S_{j}(U_I)=\sigma_{j}(U_I)=0$ if $j<0$ or $j>\left|U_I\right|$.

 TGRS codes are generalizations of GRS codes and were first introduced in~\cite{beelen2017twisted}.

\begin{definition}[\cite{beelen2017twisted}]
Let $\ell,k,n$ be positive integers  with $\ell \le k \le n \le q$, suppose that $\mathbf{h}=(h_1,h_2,\ldots,h_{\ell}),\mathbf{t}=(t_1,t_2,\ldots,t_{\ell})$ and 
$\boldsymbol{\eta}=(\eta_1,\eta_2,\ldots,\eta_{\ell})\in \mathbb{F}_q^{\ell}$, where $0 \le h_i \le k-1$ are distinct and $0 \le t_i \le n-k$ are also distinct.
Then
\[
\mathcal{S}
=
\left\{
\sum_{i=0}^{k-1} f_i x^i
+
\sum_{j=1}^{\ell} \eta_j f_{h_j} x^{k-1+t_j}
\;:\;
f_0,f_1,\ldots,f_{k-1}\in\mathbb{F}_q
\right\}
\]
is a $k$-dimensional subspace of $\mathbb{F}_q[x]$ over $\mathbb{F}_q$. Furthermore, let
$
\boldsymbol{\alpha}=\{\alpha_1,\alpha_2,\ldots,\alpha_n\}\subseteq \mathbb{F}_q,
$
where $\alpha_i$, $i=1,2,\ldots,n$, are distinct, and
$
\boldsymbol{v}=(v_1,v_2,\ldots,v_n)\in (\mathbb{F}_q^{*})^n.
$
The linear code
\[
\mathcal{C}
=
\left\{
\operatorname{ev}_{\boldsymbol{\alpha},\boldsymbol{v}}(f(x))
\,:\,
f(x)\in\mathcal{S}
\right\}
\]
is called a twisted generalized Reed-Solomon $(\operatorname{TGRS})$ code.
When $\boldsymbol{v}=(1,\ldots,1)$, the code is referred to as twisted Reed-Solomon $(\operatorname{TRS})$ code.

\end{definition}
Hu et al.~\cite{hu2025p} proposed a more general form of TGRS codes.
\begin{definition}[\cite{hu2025p}]\label{Def:L,P}
Let $n$, $k$, and $s$ be integers satisfying $0<k\leq n$ and $0\leq s\leq n-k$. Let $\mathcal{L}\subseteq[0,n-k-1]$ be the twist set with $s=|\mathcal{L}|$, let $\mathcal{P}\subseteq[0,k-1]$ be the position set, and let $B=[b_{i,j}]\in\mathbb{F}_q^{k\times(n-k)}$ be the coefficient matrix, where $0\leq i\leq k-1$ and $0\leq j\leq n-k-1$. Define
\begin{equation}
F(\mathcal{L},\mathcal{P},B)
=
\left\{
\sum_{i=0}^{k-1}f_i x^i
+
\sum_{i\in\mathcal{P}}f_i\sum_{j\in\mathcal{L}}b_{i,j}x^{k+j}
:
f_i\in\mathbb{F}_q,\ 0\leq i\leq k-1
\right\}.
\label{eq:LP-polynomial-space}
\end{equation}
Let $\boldsymbol{\alpha}=\{\alpha_1,\alpha_2,\ldots,\alpha_n\}\subseteq\mathbb{F}_q$ consist of pairwise distinct elements, and let $\boldsymbol{v}=(v_1,v_2,\ldots,v_n)\in(\mathbb{F}_q^*)^n$. The code
\begin{equation}
C(\mathcal{L},\mathcal{P},B)
=
\left\{
\operatorname{ev}_{\boldsymbol{\alpha},\boldsymbol{v}}(f)
=
\bigl(v_1f(a_1),\ldots,v_nf(a_n)\bigr)
:
f\in F(\mathcal{L},\mathcal{P},B)
\right\}
\label{eq:LP-TGRS-code}
\end{equation}
is called an $(\mathcal{L},\mathcal{P})$-twisted generalized Reed-Solomon code, abbreviated as an $(\mathcal{L},\mathcal{P})$-TGRS code. When $\boldsymbol{v}=(1,\ldots,1)$, it is called an $(\mathcal{L},\mathcal{P})$-twisted Reed-Solomon code, abbreviated as an $(\mathcal{L},\mathcal{P})$-TRS code.
\end{definition}

Roth and Lempel proposed Roth-Lempel (RL) codes in 1989~\cite{roth1989construction}.

\begin{definition}~\cite{roth1989construction}\label{def:ordinary-RL-code}
	Let $3\le k<n$ and $\delta\in\mathbb{F}_{q}$. Let $\boldsymbol{\alpha}=\{\alpha_1,\ldots,\alpha_n\}\subseteq\mathbb F_q$ with $\alpha_i\neq\alpha_j$ for $i\neq j$ and
$\boldsymbol{v}=(v_1,\ldots,v_n)\in (\mathbb{F}_q^*)^n$.
    The generalized Roth-Lempel code of dimension $k$ associated with $\boldsymbol{\alpha},\boldsymbol{v}$ and $\delta$ is defined by
	\[
	\operatorname{GRL}_{k,n+2}(\boldsymbol{\alpha},\boldsymbol{v},\delta)=\left\{\bigl(v_1 f(\alpha_1),\ldots,v_nf(\alpha_n),f_{k-1},f_{k-2}+\delta f_{k-1}\bigr)\mid f(x)=\sum_{i=0}^{k-1}f_i x^i,\ f_i\in\mathbb F_q\right\}.
	\]
	By definition, $\operatorname{GRL}_{k,n+2}(\boldsymbol{\alpha},\boldsymbol{v},\delta)$ is an $[n+2,k]$ linear code over $\mathbb F_q$. When $\boldsymbol{v}=(1,\ldots,1)$, the code is referred to as $\operatorname{RL}_{k,n+2}(\boldsymbol{\alpha},\delta)$.
\end{definition}
  In what follows, we keep the same Roth-Lempel extension coordinates, but replace the ordinary polynomial space by a  twisted polynomial space.

 \begin{definition}\label{def:TRL-code}
 With the notation as in Definition~\ref{Def:L,P}, let $\boldsymbol{\alpha}=\{\alpha_1,\alpha_2,\cdots,\alpha_n\}\subseteq \mathbb{F}_q$ with distinct $\alpha_1,\ldots,\alpha_n$ and
$\boldsymbol{v}=(v_1,\ldots,v_n)\in (\mathbb{F}_q^*)^n$. Let $\delta\in\mathbb{F}_{q}$,
then the twisted generalized 
Roth-Lempel $(\operatorname{TGRL})$ code is defined by
\begin{equation}
\operatorname{TGRL}(\boldsymbol{\alpha},\boldsymbol{v},\delta,\mathcal{L},\mathcal{P},B)
=\left\{
(v_1f(\alpha_1),\ldots,v_nf(\alpha_n),f_{k-1},f_{k-2}+\delta\cdot f_{k-1})
:
f(x)\in F(\mathcal{L},\mathcal{P},B)
\right\}.
\end{equation}
We call it the $(\mathcal{L},\mathcal{P})$-$\operatorname{TGRL}$ code. When $\boldsymbol{v}=(1,\ldots,1)$, the code is abbreviated as $\operatorname{TRL}(\boldsymbol{\alpha},\delta,\mathcal{L},\mathcal{P},B)$ and is called the $(\mathcal{L},\mathcal{P})$-$\operatorname{TRL}$ code.
 \end{definition}

In this paper, we shall consider the cases $\mathcal{L}=\{0\}$ and $\mathcal{P}=\{\ell\}$, where $0\leq\ell\leq k-1$. Specifically, let integers $\ell,k,n$ be such that $0\leq\ell\leq k-1\leq n$. For any $\eta\in\mathbb{F}_{q}$ denote by
\begin{equation*}
    S_{k,\ell,\eta}=\left\{\sum\limits_{0\leq i\leq k-1\atop i\neq\ell}f_{i}x^i+f_{\ell}\left(x^{\ell}+\eta x^k\right):f_{0},f_{1},\cdots,f_{k-1}\in\mathbb{F}_{q}\right\}.
\end{equation*}
Let $\delta\in\mathbb{F}_{q}$,
for any $\boldsymbol{\alpha}=\{\alpha_{1},\cdots,\alpha_{n}\}\subseteq\mathbb{F}_{q}$ with distinct $\alpha_1,\cdots,\alpha_n$ and $\boldsymbol{v}=(v_1,\cdots,v_n)\in(\mathbb{F}_{q}^{*})^n$,
we will focus on the following $\operatorname{TGRL}$ codes:
\begin{equation*}
    \operatorname{TGRL}_{k,n+2}(\boldsymbol{\alpha},\boldsymbol{v},\eta,\delta,\ell)=\left\{(v_1f(\alpha_{1}),\cdots,v_nf(\alpha_{n}),f_{k-1},f_{k-2}+\delta\cdot f_{k-1}):f\in S_{k,\ell,\eta}\right\}.
\end{equation*}
When $\boldsymbol{v}=(1,\ldots,1)$, the code is abbreviated as $\operatorname{TRL}_{k,n+2}(\boldsymbol{\alpha},\eta,\delta,\ell)$  and is called the $\operatorname{TRL}$ code.

 Obviously, $\operatorname{TGRL}_{k,n+2}(\boldsymbol{\alpha},\boldsymbol{v},\eta,\delta,\ell)$ is a $[n+2,k]$ linear code. Moreover, the generator matrix of this code is
\small{ \[
 G^{\operatorname{TGRL}}_{k,n+2}(\boldsymbol{\alpha},\boldsymbol{v},\eta,\delta,\ell)=
 \begin{pmatrix}
 v_1 & v_2 & \cdots & v_n & 0 & 0\\
 v_1\alpha_1 & v_2\alpha_2 & \cdots & v_n\alpha_n & 0 & 0\\
  \vdots & \vdots &  \vdots&\vdots & \vdots & \vdots\\
 v_1\alpha_1^{\ell-1} & v_2\alpha_2^{\ell-1} & \cdots & v_n\alpha_n^{\ell-1} & 0 & 0\\
  v_1\left(\alpha_1^{\ell}+\eta\alpha_{1}^k\right) &v_2\left(\alpha_2^{\ell}+\eta\alpha_{2}^{k}\right)&\cdots&v_n\left(\alpha_n^{\ell}+\eta\alpha_{n}^k\right) &0&0\\
 v_1\alpha_1^{\ell+1} & v_2\alpha_2^{\ell+1} & \cdots & v_n\alpha_n^{\ell+1} & 0 & 0\\
 \vdots & \vdots &\vdots  & \vdots & \vdots & \vdots\\
  v_1\alpha_1^{k-3} & v_2\alpha_2^{k-3} & \cdots & v_n\alpha_n^{k-3} & 0 & 0\\
 v_1\alpha_1^{k-2} & v_2\alpha_2^{k-2} & \cdots & v_n\alpha_n^{k-2} & 0 & 1\\
 v_1\alpha_1^{k-1} & v_2\alpha_2^{k-1} & \cdots & v_n\alpha_n^{k-1} & 1 & \delta\\
 \end{pmatrix}
 \]}
 for all $0\leq\ell\leq k-3$,
 \small{\[
 G^{\operatorname{TGRL}}_{k,n+2}(\boldsymbol{\alpha},\boldsymbol{v},\eta,\delta,k-2)=
 \begin{pmatrix}
 v_1 & v_2 & \cdots & v_n & 0 & 0\\
 v_1\alpha_1 & v_2\alpha_2 & \cdots & v_n\alpha_n & 0 & 0\\
  \vdots & \vdots &  \vdots&\vdots & \vdots & \vdots\\
  v_1\alpha_1^{k-3} & v_2\alpha_2^{k-3} & \cdots & v_n\alpha_n^{k-3} & 0 & 0\\
v_1\left(\alpha_1^{k-2}+\eta\alpha_{1}^k\right) &v_2\left(\alpha_2^{k-2}+\eta\alpha_{2}^{k}\right)&\cdots&v_n\left(\alpha_n^{k-2}+\eta\alpha_{n}^k\right) &0&1\\
 v_1\alpha_1^{k-1} & v_2\alpha_2^{k-1} & \cdots & v_n\alpha_n^{k-1} & 1 & \delta\\
 \end{pmatrix}
 \]}
 and
\small{  \[
 G^{\operatorname{TGRL}}_{k,n+2}(\boldsymbol{\alpha},\boldsymbol{v},\eta,\delta,k-1)=
 \begin{pmatrix}
 v_1 & v_2 & \cdots & v_n & 0 & 0\\
 v_1\alpha_1 & v_2\alpha_2 & \cdots & v_n\alpha_n & 0 & 0\\
  \vdots & \vdots &  \vdots&\vdots & \vdots & \vdots\\
  v_1\alpha_1^{k-3} & v_2\alpha_2^{k-3} & \cdots & v_n\alpha_n^{k-3} & 0 & 0\\
   v_1\alpha_1^{k-2} & v_2\alpha_2^{k-2} & \cdots & v_n\alpha_n^{k-2} & 0& 1\\
v_1\left(\alpha_1^{k-1}+\eta\alpha_{1}^k\right) &v_2\left(\alpha_2^{k-1}+\eta\alpha_{2}^{k}\right)&\cdots&v_n\left(\alpha_n^{k-1}+\eta\alpha_{n}^k\right) &1&\delta\\
 \end{pmatrix}.
 \]}

The Schur product of linear codes over a finite field has emerged as a fundamental operation in both classical and quantum coding theory~\cite{randriambololona2015products}. It is widely used in the construction of code-equivalence distinguishers in coding theory and code-based cryptography.
\begin{definition}\label{schur product}
For $\boldsymbol{x}=(x_1,x_2,\cdots ,x_n)$, $\boldsymbol{y}=(y_1,y_2,\cdots ,y_n) \in {\mathbb{F}_q ^n}$, the \textit{Schur product} of $\boldsymbol{x}$ and $\boldsymbol{y}$ is defined as $\boldsymbol{x}*\boldsymbol{y}:=\left( {{x_1}{y_1}, \ldots ,{x_n}{y_n}} \right)$. The Schur product of two linear codes $\mathcal{C}_1,\mathcal{C}_2 \subseteq \mathbb{F}_q^n$ is defined as
\[\mathcal{C}_1*\mathcal{C}_2: =\operatorname{Span}_{\mathbb{F}_{q}}\left\langle {{\textit{\textbf{c}}_1}*{\textit{\textbf{c}}_2}:{\textit{\textbf{c}}_1} \in \mathcal{C}_1,{\textit{\textbf{c}}_2} \in \mathcal{C}_2} \right\rangle.\]
In particular, the Schur product of $\mathcal{C}$ with itself is denoted by $\mathcal{C}^2$ and is often called the Schur square of $\mathcal{C}$.
\end{definition}

The notion of equivalence for linear codes is defined as follows.
\begin{definition}[\cite{beelen2017twisted}]
  Let $\mathcal{C}$,~$\mathcal{D}$ be $[n,k]$ linear codes over $\mathbb{F}_q$. We say that $\mathcal{C}$ and $\mathcal{D}$ are \textit{equivalent} if there is a permutation $\pi \in S_n$ and $\boldsymbol{v}$=$(v_1,v_2,\cdots ,v_n)\in {\left( {\mathbb{F}_q^ * } \right)^n}$ such that $\mathcal{C}={\varphi _{\pi ,\textit{\textbf{v}}}}(\mathcal{D})$ where 
\[{\varphi _{\pi ,v}}:\mathbb{F}_q^n \to \mathbb{F}_q^n,\left( {{c_1}, \ldots ,{c_n}} \right) \mapsto \left( {{v_1}{c_{\pi \left( 1 \right)}}, \ldots ,{v_n}{c_{\pi \left( n \right)}}} \right)\]
is the Hamming-metric isometry of $\mathbb{F}_{q}^n$.
\end{definition}

The following proposition determines the Schur square of a GRS code.

\begin{proposition}[\cite{couvreur2014distinguisher}]\label{Prop:schur of RS}
  If $k \leq \frac{n}{2}$, then
$GRS_{k}(\mathcal{A},\boldsymbol{v})^2 = GRS_{2k-1}(\mathcal{A},\boldsymbol{v}^2)$.

\end{proposition}

\begin{remark}
If two $[n,k]$-codes $\mathcal{C}_{1}$ and $\mathcal{C}_{2}$ over $\mathbb{F}_{q}$ are equivalent, then $\mathcal{C}_{1}^2$ and $\mathcal{C}_{2}^2$ are equivalent. Hence, if an $[n,k,n-k+1]$ code $\mathcal{C}$ with $k\leq\frac{n}{2}$ satisfies $\dim(\mathcal{C}^2)\neq 2k-1$, then it is not equivalent to  any $\operatorname{RS}/\operatorname{GRS}$ code. In this paper, we call such codes $\operatorname{non}$-$\operatorname{RS}$ $\operatorname{MDS}$ codes.
\end{remark}

\section{The MDS property of the twisted Roth-Lempel code}
In this section, we investigate the distance and structural properties of the codes $\operatorname{TRL}_{k,n+2}(\boldsymbol{\alpha},\eta,\delta,\ell)$. We first determine the range of their minimum distances and characterize the cases in which the minimum distance is $n-k$ or $n-k+1$. We then establish necessary and sufficient conditions for these codes to be MDS, analyze their Schur squares to distinguish them from RS and classical $\operatorname{RL}$ codes, and finally present explicit constructions of non-RS MDS $\operatorname{TRL}$ codes.
First, we determine the range of the minimum distance of the code $\operatorname{TRL}_{k,n+2}(\boldsymbol{\alpha},\eta,\delta,\ell)$.

\begin{lemma}\label{Lem:n-k+1 le d le n-k+3}
   Let $n,k,\ell$ be integers satisfying $3\le k<n$ and $0\le \ell\le k-1$. Let $\boldsymbol{\alpha}=\{\alpha_1,\ldots,\alpha_n\}\subseteq\mathbb F_q$ with $\alpha_i\neq\alpha_j$ for $i\neq j$.  Let $\eta,\delta\in\mathbb F_q^{*}$, then the minimum distance $d$  of the  code $\operatorname{TRL}_{k,n+2}(\boldsymbol{\alpha},\eta,\delta,\ell)$  satisfies 
   \[
   n-k\leq d\leq n-k+3\quad\mbox{for all}\quad 0\leq\ell\leq k-3
   \]
   and 
   \[
    n-k+1\leq d\leq n-k+3\quad\mbox{for all}\quad k-2\leq\ell\leq k-1.
   \]
\end{lemma}
\begin{proof}
Let $\mathcal{C} = \operatorname{TRL}_{k,n+2}(\boldsymbol{\alpha},\eta, \delta, \ell)$. By the Singleton bound, we know that $d(\mathcal{C}) \leq n-k+3$. For $0 \leq \ell \leq k-1$, we classify the code $\mathcal{C}$ into the following four cases according to whether the last two coordinates of the codeword are zero or not:
\begin{align*}
T_1^{(\ell)} &= \{(f(\alpha_1),\cdots,f(\alpha_n), f_{k-1}, f_{k-2}+\delta f_{k-1}) : f(x)\in S_{k,\ell,\eta}, f_{k-1}=f_{k-2}=0\},\\
T_2^{(\ell)} &= \{(f(\alpha_1),\cdots,f(\alpha_n), f_{k-1}, f_{k-2}+\delta f_{k-1}) : f(x)\in S_{k,\ell,\eta}, f_{k-1}=0, f_{k-2}\neq 0\},\\
T_3^{(\ell)} &= \{(f(\alpha_1),\cdots,f(\alpha_n), f_{k-1}, f_{k-2}+\delta f_{k-1}) : f(x)\in S_{k,\ell,\eta}, f_{k-1}\neq 0, f_{k-2}=-\delta f_{k-1}\},\\
T_4^{(\ell)} &= \{(f(\alpha_1),\cdots,f(\alpha_n), f_{k-1}, f_{k-2}+\delta f_{k-1}): f(x)\in S_{k,\ell,\eta}, f_{k-1}\neq 0, f_{k-2}\neq -\delta f_{k-1}\}.
\end{align*}
For $\ell=k-2$, we shall consider four cases separately:
\begin{itemize}
    \item [(i)] If $\mathbf{c}=(c_1,c_2,\cdots,c_{n+2})\in T_1^{(k-2)}\setminus\{\boldsymbol{0}\}$, then $c_{n+1}=c_{n+2}=0$ and there exists a polynomial $f(x)$ of degree at most $k-3$ such that $f(\alpha_i)=c_i$ for all $1\leq i\leq n$. Thus,
\[
\#\{1\leq i\leq n: c_i=0\} \leq \deg(f(x)) \leq k-3.
\]
Therefore,
\[
\min_{\mathbf{c}\in T_1^{(k-2)}\setminus\{\boldsymbol{0}\}}\#\{1\leq i\leq n+2: c_i\neq 0\} = n-\max_{\mathbf{c}\in T_1^{(k-2)}\setminus\{\boldsymbol{0}\}}\#\{1\leq i\leq n: c_i=0\} \geq n-k+3.
\]
\item [(ii)] If $\mathbf{c}=(c_1,c_2,\cdots,c_{n+2})\in T_2^{(k-2)}\setminus\{\boldsymbol{0}\}$, then $c_{n+1}=0,c_{n+2}\neq 0$ and  there exists a
\[
f(x)=\sum_{\substack{0\leq i\leq k-3}} f_i x^i+f_{k-2}(x^{k-2}+\eta x^k)\in \mathcal{S}_{k,k-2,\eta}
\]
such that $c_i=f(\alpha_i)$ for all $1\leq i\leq n$, where $f_{k-2}\neq 0$.
Since the degree of $f(x)$ is at most $k$, we have
\[
\#\{1\leq i\leq n: c_i=0\} \leq \deg(f(x)) \leq k.
\]
Hence,
\[
\min_{\mathbf{c}\in T_2^{(k-2)}\setminus\{\boldsymbol{0}\}}\#\{1\leq i\leq n+2: c_i\neq 0\} = 1+n-\max_{\mathbf{c}\in T_2^{(k-2)}\setminus\{\boldsymbol{0}\}}\#\{1\leq i\leq n: c_i=0\} \geq n-k+1.
\]

\item [(iii)] If $\boldsymbol{c}=(c_1,c_2,\cdots,c_{n+2})\in T_3^{(k-2)}\setminus\{\boldsymbol{0}\}$, then $c_{n+1}\neq 0,c_{n+2}=0$ and  there exists a
\[
f(x)=\sum_{\substack{0\leq i\leq k-3}} f_i x^i-f_{k-1}\delta(x^{k-2}+\eta x^k)+ f_{k-1}x^{k-1}\in \mathcal{S}_{k,k-2,\eta}
\]
such that $c_i=f(\alpha_i)$ for all $1\leq i\leq n$, where $f_{k-1}\neq 0$.
Since the degree of $f(x)$ is at most $k$, we have
\[
\#\{1\leq i\leq n: c_i=0\} \leq \deg(f(x)) \leq k.
\]
Hence,
\begin{equation*}
    \begin{aligned}
        \min_{\mathbf{c}\in T_3^{(k-2)}\setminus\{\boldsymbol{0}\}}\#\{1\leq i\leq n+2: c_i\neq 0\} &=1+\min_{\mathbf{c}\in T_3^{(k-2)}\setminus\{\boldsymbol{0}\}}\#\{1\leq i\leq n: c_i\neq 0\}\\
        &= 1+n-\max_{\mathbf{c}\in T_3^{(k-2)}\setminus\{\boldsymbol{0}\}}\#\{1\leq i\leq n: c_i=0\} \geq n-k+1.
    \end{aligned}
\end{equation*}
\item [(iv)]  If $\boldsymbol{c}=(c_1,c_2,\cdots,c_{n+2})\in T_4^{(k-2)}\setminus\{\boldsymbol{0}\}$, then $c_{n+1},c_{n+2}\neq 0$ and  there exists a
\[
f(x)=\sum_{\substack{0\leq i\leq k-1\\ i\neq k-2}} f_i x^i+f_{k-2}(x^{k-2}+\eta x^k)\in \mathcal{S}_{k,k-2,\eta}
\]
such that $c_i=f(\alpha_i)$ for all $1\leq i\leq n$, where $f_{k-1}\neq 0,f_{k-2}\neq -f_{k-1}\delta$.
Since the degree of $f(x)$ is at most $k$, we have
\[
\#\{1\leq i\leq n: c_i=0\} \leq \deg(f(x)) \leq k.
\]
Hence,
\begin{equation*}
    \begin{aligned}
        \min_{\mathbf{c}\in T_4^{(k-2)}\setminus\{\boldsymbol{0}\}}\#\{1\leq i\leq n+2: c_i\neq 0\} &=2+\min_{\mathbf{c}\in T_4^{(k-2)}\setminus\{\boldsymbol{0}\}}\#\{1\leq i\leq n: c_i\neq 0\}\\
        &= 2+n-\max_{\mathbf{c}\in T_4^{(k-2)}\setminus\{\boldsymbol{0}\}}\#\{1\leq i\leq n: c_i=0\} \geq n-k+2.
    \end{aligned}
\end{equation*}
\end{itemize}


When $\ell=k-1$, analogous to the case $\ell= k-2$, we obtain four cases concerning the lower bound on the minimum distance:
\begin{itemize}
    \item [(i)] $\operatorname{min}\left\{\operatorname{wt}(\boldsymbol{c}):\boldsymbol{c}\in T_1^{(k-1)}\backslash\{\boldsymbol{0}\}\right\}=n-\max\limits_{\boldsymbol{c}\in T_1^{(k-1)}\backslash\{\boldsymbol{0}\}}\#\left\{1\leq i\leq n:c_{i}=0\right\}\geq n-k+3$.
    \item [(ii)] $\operatorname{min}\left\{\operatorname{wt}(\boldsymbol{c}):\boldsymbol{c}\in T_2^{(k-1)}\backslash\{\boldsymbol{0}\}\right\}=n+1-\max\limits_{\boldsymbol{c}\in T_2^{(k-1)}\backslash\{\boldsymbol{0}\}}\#\left\{1\leq i\leq n:c_{i}=0\right\}\geq n-k+3$.
    \item [(iii)] $\operatorname{min}\left\{\operatorname{wt}(\boldsymbol{c}):\boldsymbol{c}\in T_3^{(k-1)}\backslash\{\boldsymbol{0}\}\right\}=n+1-\max\limits_{\boldsymbol{c}\in T_3^{(k-1)}\backslash\{\boldsymbol{0}\}}\#\left\{1\leq i\leq n:c_{i}=0\right\}\geq n-k+1$.
     \item [(iv)] $\operatorname{min}\left\{\operatorname{wt}(\boldsymbol{c}):\boldsymbol{c}\in T_4^{(k-1)}\backslash\{\boldsymbol{0}\}\right\}=n+2-\max\limits_{\boldsymbol{c}\in T_4^{(k-1)}\backslash\{\boldsymbol{0}\}}\#\left\{1\leq i\leq n:c_{i}=0\right\}\geq n-k+2$.
\end{itemize}

When $0\leq\ell\leq k-3$, analogous to the case $\ell= k-2$, we obtain four cases concerning the lower bound on the minimum distance:
\begin{itemize}
    \item [(i)] $\operatorname{min}\left\{\operatorname{wt}(\boldsymbol{c}):\boldsymbol{c}\in T_1^{(\ell)}\backslash\{\boldsymbol{0}\}\right\}=n-\max\limits_{\boldsymbol{c}\in T_1^{(\ell)}\backslash\{\boldsymbol{0}\}}\#\left\{1\leq i\leq n:c_{i}=0\right\}\geq n-k$.
    \item [(ii)] $\operatorname{min}\left\{\operatorname{wt}(\boldsymbol{c}):\boldsymbol{c}\in T_2^{(\ell)}\backslash\{\boldsymbol{0}\}\right\}=n+1-\max\limits_{\boldsymbol{c}\in T_2^{(\ell)}\backslash\{\boldsymbol{0}\}}\#\left\{1\leq i\leq n:c_{i}=0\right\}\geq n-k+1$.
    \item [(iii)] $\operatorname{min}\left\{\operatorname{wt}(\boldsymbol{c}):\boldsymbol{c}\in T_3^{(\ell)}\backslash\{\boldsymbol{0}\}\right\}=n+1-\max\limits_{\boldsymbol{c}\in T_3^{(\ell)}\backslash\{\boldsymbol{0}\}}\#\left\{1\leq i\leq n:c_{i}=0\right\}\geq n-k+1$.
     \item [(iv)] $\operatorname{min}\left\{\operatorname{wt}(\boldsymbol{c}):\boldsymbol{c}\in T_4^{(\ell)}\backslash\{\boldsymbol{0}\}\right\}=n+2-\max\limits_{\boldsymbol{c}\in T_4^{(\ell)}\backslash\{\boldsymbol{0}\}}\#\left\{1\leq i\leq n:c_{i}=0\right\}\geq n-k+2$.
\end{itemize}

In total, we see that the minimum distance of the code $\mathcal{C}$  satisfies $n-k\leq d\leq n-k+3$ for $0\leq\ell\leq k-3$ and 
$n-k+1\leq d\leq n-k+3$ for $k-2\leq\ell\leq k-1$.
\end{proof}

Next, we  present necessary and sufficient conditions for the minimum distance of the code \small{$\operatorname{TRL}_{k,n+2}(\boldsymbol{\alpha},\eta,\delta,\ell)$}  to be $n-k$ for all $0\leq\ell\leq k-3$.

\begin{lemma}\label{Lem:d=n-k}
     Let $n,k,\ell$ be integers satisfying $3\le k<n$ and $0\le \ell\le k-3$. Let $\boldsymbol{\alpha}=\{\alpha_1,\ldots,\alpha_n\}\subseteq\mathbb F_q$ with $\alpha_i\neq\alpha_j$ for $i\neq j$.  Let $\eta,\delta\in\mathbb{F}_{q}^{*}$, then the  minimum distance of the code $\operatorname{TRL}_{k,n+2}(\boldsymbol{\alpha},\eta,\delta,\ell)$ is $n-k$ if and only if there exists a $k$-subset $I\subseteq [n]$ such that $\sigma_{1}(\boldsymbol{\alpha}_{I})=\sigma_{2}(\boldsymbol{\alpha}_{I})=0$ and $\eta\cdot\sigma_{k-\ell}(\boldsymbol{\alpha}_{I})=1$.

\end{lemma}

\begin{proof}
    Let $\mathcal{C} = \operatorname{TRL}_{k,n+2}(\boldsymbol{\alpha}, \eta, \delta, \ell)$. It follows from the analysis of Lemma~\ref{Lem:n-k+1 le d le n-k+3} that for a given codeword $\boldsymbol{c}=(c_{1},c_{2},\cdots,c_{n+2})\in\mathcal{C}$ , the case $\operatorname{wt}(\boldsymbol{c})=n-k$  can only occur in Case (i) for all $0\leq\ell\leq k-3$. Thus, $\operatorname{wt}(\boldsymbol{c})=n-k$ 
     if and only if there exists a
\[
f(x) = \sum_{\substack{0 \leq i \leq k-3 \\ i \neq \ell}} f_i x^i + f_{\ell}(x^\ell + \eta x^k) \in S_{k,\ell,\eta}
\]
such that
\[
\mathbf{c} = (f(\alpha_1), f(\alpha_2), \ldots, f(\alpha_n), 0, 0)
\qquad \mbox{and}\ \qquad
\#\{1 \leq i \leq n : f(\alpha_i) = 0\} = k,
\]
where $f_{\ell}\neq 0$.
 Since $\deg(f(x)) = k$ and $\#\{1 \leq i \leq n : f(\alpha_i) = 0\} = k$, there exists a $k$-subset $I \subseteq [n]$ such that
\[
f(x) = f_{\ell} \eta \prod_{j=1}^{k}(x - \alpha_{i_j}) = f_{\ell} \eta \sum_{j=0}^{k} \sigma_j(\boldsymbol{\alpha}_{I}) x^{k-j}.
\]
Comparing the coefficients of $x^{k-1},x^{k-2}$ and $x^\ell$ on both sides, we obtain $$
f_{\ell}\eta\sigma_{1}(\boldsymbol{\alpha}_{I})=0,\quad
f_{\ell}\eta\sigma_{2}(\boldsymbol{\alpha}_{I})=0\quad\mbox{and}\quad 
f_{\ell}\eta\sigma_{k-\ell}(\boldsymbol{\alpha}_{I})=f_{\ell}.$$
Hence, $ \sigma_1(\boldsymbol{\alpha}_{I})=\sigma_2(\boldsymbol{\alpha}_{I})=0$ and $\eta \sigma_{k-\ell}(\boldsymbol{\alpha}_{I}) = 1$. Moreover, the minimum distance $d$ of the code $\mathcal{C}$ satisfies $n-k \leq d \leq n-k+3$. In summary, the minimum distance of the code $\mathcal{C}$ is $n-k$ if and only if there exists a $k$-subset $I\subseteq [n]$ such that $\sigma_1(\boldsymbol{\alpha}_{I})=\sigma_{2}(\boldsymbol{\alpha}_I)=0$ and $\eta \sigma_{k-\ell}(\boldsymbol{\alpha}_{I}) = 1$.
\end{proof}

The following example demonstrates the existence of a code $\operatorname{TRL}_{k,n+2}(\boldsymbol{\alpha},\eta,\delta,\ell)$  with minimum distance $n-k$ for some $0\leq\ell\leq k-3$.

\begin{example}
    Let $n=11,k=7$ and $\boldsymbol{\alpha}=\{22, 7, 11, 3, 14, 27, 0, 25, 8, 13, 21 \}\subseteq\mathbb{F}_{31}$. Choose $\eta=30,\delta=1,\ell=3$. Then we have
    \[
    \addtocounter{MaxMatrixCols}{13}
    G=\begin{bmatrix}
 1 & 1 & 1 & 1 & 1 & 1 & 1 & 1 & 1 & 1 & 1 & 0 & 0\\
22 & 7& 11 & 3 &14 &27 & 0 &25&  8& 13& 21&  0&  0\\
19& 18& 28&  9& 10& 16&  0&  5&  2 &14 & 7 & 0&  0\\
25&  5& 16& 10& 28& 14&  0&  7& 14&  5& 12&  0&  0\\
20& 14&  9 &19&  7&  8&  0 &25&  4& 10& 18&  0&  0\\
 6&  5&  6& 26&  5& 30&  0&  5&  1&  6&  6 & 0 & 1\\
 8&  4&  4& 16&  8&  4&  0&  1&  8& 16&  2&  1&  1
    \end{bmatrix}.
    \]
On the one hand,  Magma verifies that there exists a $7$-subset $I= \{2, 3, 5, 6, 7, 10,11\}\subseteq [11]$ such that $\sigma_{1}(\boldsymbol{\alpha}_{I})=\sigma_{2}(\boldsymbol{\alpha}_{I})=0$ and $\eta\cdot\sigma_{4}(\boldsymbol{\alpha}_{I})=1$. Thus, the minimum distance of the code $\operatorname{TRL}_{7,13}(\boldsymbol{\alpha},30,1,3)$ is $n-k=4$.
On the other hand, from the given generator matrix, it can be directly verified  that this code is a linear code with parameters $[13,7,4]$ over $\mathbb{F}_{31}$ by using Magma. 
\end{example}

The following lemmas provide the necessary and sufficient conditions under which the code $\operatorname{TRL}_{k,n+2}(\boldsymbol{\alpha},\eta,\delta,\ell)$  has minimum distance $n-k+1$ for $0\leq\ell\leq k-1$.

\begin{lemma}\label{Lem:d=n-k+1 and ell=k-2,k-1}
     Let $n,k$ be integers satisfying $3\le k<n$. Let $\boldsymbol{\alpha}=\{\alpha_1,\ldots,\alpha_n\}\subseteq\mathbb F_q$ with $\alpha_i\neq\alpha_j$ for $i\neq j$.  Let $\eta,\delta\in\mathbb{F}_{q}^{*}$. Then the minimum distance of the code $\operatorname{TRL}_{k,n+2}(\boldsymbol{\alpha},\eta,\delta,k-2)$ is $n-k+1$ if and only if there exists a $k$-subset $I\subseteq [n]$ such that $\sigma_{1}(\boldsymbol{\alpha}_{I})=0,\eta\cdot\sigma_{2}(\boldsymbol{\alpha}_{I})=1$ or there exists a $k$-subset $J\subseteq [n]$ such that $\delta\eta\sigma_{1}(\boldsymbol{\alpha}_{J})=-1,\eta\cdot\sigma_{2}(\boldsymbol{\alpha}_{J})=1$.  The minimum distance of the code $\operatorname{TRL}_{k,n+2}(\boldsymbol{\alpha},\eta,\delta,k-1)$ is $n-k+1$ if and only if there exists a $k$-subset $I\subseteq [n]$ such that $\sigma_{1}(\boldsymbol{\alpha}_{I})=\eta^{-1}$ and $\sigma_{2}(\boldsymbol{\alpha}_{I})=-\delta\eta^{-1}$.

\end{lemma}

\begin{proof}
    Let $\mathcal{C}_{\ell} = \operatorname{TRL}_{k,n+2}(\boldsymbol{\alpha}, \eta, \delta, \ell)$ for $k-2\leq\ell\leq k-1$. It follows from the analysis of Lemma~\ref{Lem:n-k+1 le d le n-k+3} that for a given codeword $\boldsymbol{c}=(c_{1},c_{2},\cdots,c_{n+2})\in\mathcal{C}_{k-2}$ , the case $\operatorname{wt}(\boldsymbol{c})=n-k+1$  can only occur in Case (ii) or Case (iii). Thus, $\operatorname{wt}(\boldsymbol{c})=n-k+1$
     if and only if 
     \begin{itemize}
         \item [(1)] There exists a
         \[
f_1(x) = \sum_{0 \leq i \leq k-3 } f_{1,i} x^i + f_{1,k-2}(x^{k-2} + \eta x^k) \in S_{k,k-2,\eta}
\]
such that
\[
\mathbf{c} = (f_1(\alpha_1), f_1(\alpha_2), \cdots, f_1(\alpha_n), 0, c_{n+2})
\qquad \mbox{and}\ \qquad
\#\{1 \leq i \leq n : f_1(\alpha_i) = 0\} = k,
\]
where $c_{n+2}\neq 0$ and $f_{1,k-2}\neq 0$.
\item [(2)] There exists a
         \[
f_2(x) = \sum_{0 \leq i \leq k-3 } f_{2,i} x^i+f_{2,k-1}x^{k-1}-f_{2,k-1}\delta(x^{k-2} + \eta x^k) \in S_{k,k-2,\eta}
\]
such that
\[
\mathbf{c} = (f_2(\alpha_1), f_2(\alpha_2), \cdots, f_2(\alpha_n), c_{n+1}, 0)
\qquad \mbox{and}\ \qquad
\#\{1 \leq i \leq n : f_2(\alpha_i) = 0\} = k,
\]
where $c_{n+1}\neq 0$ and $f_{2,k-1}\neq 0$.
     
\end{itemize}

For the former, since  $\deg(f_1(x)) = k$ and $\#\{1 \leq i \leq n : f_1(\alpha_i) = 0\} = k$, there exists a $k$-subset $I\subseteq [n]$ such that
\[
f_1(x) = f_{1,k-2}\eta \prod_{i\in I}(x-\alpha_{i}) = f_{1,k-2}\eta \sum_{j=0}^{k} \sigma_j(\boldsymbol{\alpha}_{I}) x^{k-j}.
\]
Comparing the coefficients of $x^{k-1}$ and $x^{k-2}$ on both sides, we obtain $\sigma_{1}(\boldsymbol{\alpha}_{I})=0$ and $\eta\sigma_{2}(\boldsymbol{\alpha}_{I})=1$.

For the latter, since $\deg(f_2(x)) = k$ and $\#\{1 \leq i \leq n : f_2(\alpha_i) = 0\} = k$, there exists a $k$-subset $J\subseteq [n]$ such that
\[
f_{2}(x)=-f_{2,k-1}\delta\eta\prod\limits_{j\in J}(x-\alpha_{j})=-f_{2,k-1}\delta\eta\sum\limits_{j=0}^k\sigma_{j}(\boldsymbol{\alpha}_{J})x^{k-j}.
\]
Comparing the coefficients of $x^{k-1}$ and $x^{k-2}$ on both sides, we obtain $\delta\eta\sigma_{1}(\boldsymbol{\alpha}_{J})=-1$ and $\eta\cdot\sigma_{2}(\boldsymbol{\alpha}_{J})=1$.

 Moreover, the minimum distance $d$ of the code $\mathcal{C}_{k-2}$ satisfies $n-k+1 \leq d \leq n-k+3$. In summary, the minimum distance of the code $\operatorname{TRL}_{k,n+2}(\boldsymbol{\alpha},\eta,\delta,k-2)$ is $n-k+1$ if and only if there exists  a $k$-subset $I\subseteq [n]$ such that $\sigma_{1}(\boldsymbol{\alpha}_{I})=0,\eta\cdot\sigma_{2}(\boldsymbol{\alpha}_{I})=1$ or there exists a $k$-subset $J\subseteq [n]$ such that $\delta\eta\sigma_{1}(\boldsymbol{\alpha}_{J})=-1$ and $\eta\cdot\sigma_{2}(\boldsymbol{\alpha}_{J})=1$. 

It follows from the analysis of Lemma~\ref{Lem:n-k+1 le d le n-k+3} that for a given codeword $\boldsymbol{c}=(c_{1},c_{2},\cdots,c_{n+2})\in\mathcal{C}_{k-1}$ , the case $\operatorname{wt}(\boldsymbol{c})=n-k+1$  can only occur in Case (iii).
Similar to the above analysis, it is not difficult to prove
the minimum distance of the code $\mathcal{C}_{k-1}$ is $n-k+1$ if and only if there exists a $k$-subset $I\subseteq [n]$ such that $\sigma_{1}(\boldsymbol{\alpha}_{I})=\eta^{-1}$ and $\sigma_{2}(\boldsymbol{\alpha}_{I})=-\delta\eta^{-1}$.
\end{proof}

\begin{lemma}\label{Lem:d=n-k+1 and 0 le ell le k-3}
Let $n,k$ be integers satisfying $3\le k<n$ and $0\leq\ell\leq k-3$. Let $\boldsymbol{\alpha}=\{\alpha_1,\ldots,\alpha_n\}\subseteq\mathbb F_q$ with $\alpha_i\neq\alpha_j$ for $i\neq j$. Let $\eta,\delta\in\mathbb{F}_{q}^{*}$, then the minimum distance of the code $\operatorname{TRL}_{k,n+2}(\boldsymbol{\alpha},\eta,\delta,\ell)$ is $n-k+1$ if and only if
\begin{description}
    \item[\normalfont (A)] no $k$-subset $I_1\subseteq[n]$ such that $\sigma_{1}(\boldsymbol{\alpha}_{I_1})=\sigma_{2}(\boldsymbol{\alpha}_{I_1})=0$ and $\eta\sigma_{k-\ell}(\boldsymbol{\alpha}_{I_1})=1$;
    
    \item[\normalfont (B)] at least one of the following conditions holds:
    \begin{description}
        \item[\normalfont (i)] There exists a $(k-1)$-subset $I_{2}\subseteq[n]$ such that
        \[
        \sigma_{1}(\boldsymbol{\alpha}_{I_2})\notin\{\alpha_1,\cdots,\alpha_n\}\setminus\boldsymbol{\alpha}_{I_2},\ 
        \sigma_{2}(\boldsymbol{\alpha}_{I_2})=\sigma_{1}^2(\boldsymbol{\alpha}_{I_2})\ \mbox{and}\ 
        \sigma_{k-\ell}(\boldsymbol{\alpha}_{I_2})
        =\sigma_{1}(\boldsymbol{\alpha}_{I_2})
        \sigma_{k-\ell-1}(\boldsymbol{\alpha}_{I_2})+\eta^{-1};
        \]
        
        \item[\normalfont (ii)] there exists a $k$-subset $I_3\subseteq[n]$ such that $\sigma_{1}(\boldsymbol{\alpha}_{I_3})=0$ and $\eta\sigma_{k-\ell}(\boldsymbol{\alpha}_{I_3})=1$;
        
        \item[\normalfont (iii)] there exists a $k$-subset $I_4\subseteq[n]$ such that $\delta\sigma_{1}(\boldsymbol{\alpha}_{I_4})+\sigma_{2}(\boldsymbol{\alpha}_{I_4})=0$ and $\eta\sigma_{k-\ell}(\boldsymbol{\alpha}_{I_4})=1$.
    \end{description}
\end{description}
\end{lemma}

     \begin{proof}
        Let $\mathcal{C}_{\ell} =\operatorname{TRL}_{k,n+2}(\boldsymbol{\alpha}, \eta, \delta, \ell)$ for $0\leq\ell\leq k-3$. It follows from the analysis of Lemma~\ref{Lem:n-k+1 le d le n-k+3} that for a given codeword $\boldsymbol{c}=(c_{1},c_{2},\cdots,c_{n+2})\in\mathcal{C}$ , the case $\operatorname{wt}(\boldsymbol{c})=n-k+1$  can only occur in Case (i)-(iii). We shall consider the following three cases to present the necessary and sufficient conditions for the minimum distance of the code $\mathcal{C}_{\ell}$  to be $n-k+1$.
        \begin{itemize}
            \item [(1)] Since $\min\{\mathrm{wt}(\mathbf{c}) : \mathbf{c} \in T_1^{(\ell)} \setminus \{\mathbf{0}\}\} \geq n-k$, for a given $\mathbf{c} = (c_1, \ldots, c_{n+2}) \in T_1^{(\ell)}$, we have $\mathrm{wt}(\mathbf{c}) = n-k+1$, if and only if $\mathrm{wt}(\mathbf{c})> n-k$ and $\mathrm{wt}(\mathbf{c})\leq n-k+1$, if and only if the following hold:

\begin{enumerate}
\item [(a)] There does not exist a
$
f(x) = \sum\limits_{\substack{0 \leq i \leq k-3 \\ i \neq \ell}} f_i x^i + f_\ell(x^\ell + \eta x^k) \in S_{k,\ell,\eta}
$
  such that
 $$\mathbf{c} = (f(\alpha_1), f(\alpha_2), \ldots, f(\alpha_n),0,0)\quad \mbox{and}\quad 
\#\{1 \leq i \leq n : f(\alpha_i) = 0\} = k,$$
where $f_\ell \neq 0$ .
\item [(b)] There exists a
$
g(x) = \sum\limits_{0 \leq i \leq k-3\atop i\neq\ell } g_i x^i + g_\ell(x^\ell + \eta x^k) \in S_{k,\ell,\eta}
$
such that $$\mathbf{c} = (g(\alpha_1), \ldots, g(\alpha_n), 0, 0)\quad \mbox{and}\quad 
\#\{1 \leq i \leq n : g(\alpha_i) = 0\} = k-1,$$
where $g_{\ell}\neq 0$.
\end{enumerate}

For Case (a), it is equivalent to the non-existence of a $k$-subset $I_1\subseteq [n]$ such that
\[
f(x) = f_\ell \eta \prod_{i\in I}(x - \alpha_{i}) = f_\ell \eta \sum_{j=0}^{k} \sigma_{k-j}(\boldsymbol{\alpha}_{I_1}) x^j.
\]
Comparing the coefficients of $x^{k-1},x^{k-2}$ and $x^\ell$ on both sides yields $\sigma_{1}(\boldsymbol{\alpha}_{I_1})=\sigma_{2}(\boldsymbol{\alpha}_{I_1})=0$ and $\eta \sigma_{k-\ell}(\boldsymbol{\alpha}_{I_1}) = 1$.

For Case (b), it is equivalent to the existence of a $k-1$-subset $I_2 = \{i_1, i_2, \ldots, i_{k-1}\} \subseteq [n]$ and $t \notin \{\alpha_1, \alpha_2, \ldots, \alpha_n\} \setminus \boldsymbol{\alpha}_{I_2}$ such that
\[
g(x) = g_\ell \eta (x-t) \prod_{s=1}^{k-1}(x - \alpha_{i_s}) = g_\ell \eta \sum_{s=0}^{k} \sigma_{k-s}(\{\boldsymbol{\alpha}_{I_2}, t\}) x^s.
\]
Comparing the coefficients of $x^{k-1},x^{k-2}$ and $x^\ell$ on both sides yields
\[
\sigma_{1}(\{\boldsymbol{\alpha}_{I_{2}},t\})=\sigma_{2}(\{\boldsymbol{\alpha}_{I_2},t\})=0\quad\mbox{and}\quad \sigma_{k-\ell}(\{\boldsymbol{\alpha}_{I_2},t\})=\eta^{-1}.
\]

For $\sigma_{1}(\{\boldsymbol{\alpha}_{I_{2}},t\})=0$ ,since $\sigma_{1}(\{\boldsymbol{\alpha}_{I_{2}},t\})=\sigma_{1}(\boldsymbol{\alpha}_{I_2})-t$,
 it suffices to require
 \[
\sigma_{1}(\boldsymbol{\alpha}_{I_2})=t\notin\{\alpha_1,\cdots,\alpha_n\}\setminus\boldsymbol{\alpha}_{I_2}.
 \]

 For $\sigma_{2}(\{\boldsymbol{\alpha}_{I_{2}},t\})=0$ ,since $\sigma_{2}(\{\boldsymbol{\alpha}_{I_{2}},t\})=\sigma_{2}(\boldsymbol{\alpha}_{I_2})-t\sigma_{1}(\boldsymbol{\alpha}_{I_2})$,
 it suffices to require
 \[
\sigma_{2}(\boldsymbol{\alpha}_{I_2})-\sigma_{1}^2(\boldsymbol{\alpha}_{I_2})=0.\]
 For $\sigma_{k-\ell}(\{\boldsymbol{\alpha}_{I_{2}},t\})=\eta^{-1}$ ,since $\sigma_{k-\ell}(\{\boldsymbol{\alpha}_{I_{2}},t\})=\sigma_{k-\ell}(\boldsymbol{\alpha}_{I_2})-t\sigma_{k-\ell-1}(\boldsymbol{\alpha}_{I_2})$,
 it suffices to require
 \[
\sigma_{k-\ell}(\boldsymbol{\alpha}_{I_2})-\sigma_{1}(\boldsymbol{\alpha}_{I_2})\sigma_{k-\ell-1}(\boldsymbol{\alpha}_{I_2})=\eta^{-1}.\]

Thus, $\mathrm{wt}(\mathbf{c}) = n-k+1$ if and only if the following hold:
\begin{itemize}
    \item [(i)] There does not exist a $k$-subset $I_1\subseteq [n]$ such that $\sigma_{1}(\boldsymbol{\alpha}_{I_1})=\sigma_{2}(\boldsymbol{\alpha}_{I_1})=0$ and $\eta\sigma_{k-\ell}(\boldsymbol{\alpha}_{I_1})=1$.
    \item [(ii)] There exists a $k-1$-subset $I_{2}\subseteq [n]$ such that
    \[
\sigma_{1}(\boldsymbol{\alpha}_{I_2})\notin\{\alpha_1,\cdots,\alpha_n\}\setminus\boldsymbol{\alpha}_{I_2},\ 
\sigma_{2}(\boldsymbol{\alpha}_{I_2})=\sigma_{1}^2(\boldsymbol{\alpha}_{I_2})\ \mbox{and}\  
\sigma_{k-\ell}(\boldsymbol{\alpha}_{I_2})=\sigma_{1}(\boldsymbol{\alpha}_{I_2})\sigma_{k-\ell-1}(\boldsymbol{\alpha}_{I_2})+\eta^{-1}.
\]
\end{itemize}

 \item [(2)] Since $\min\{\mathrm{wt}(\mathbf{c}) : \mathbf{c} \in T_2^{(\ell)} \setminus \{\mathbf{0}\}\} \geq n-k+1$, for a given $\mathbf{c} = (c_1, \ldots, c_{n+2}) \in T_2^{(\ell)}$, we have $\mathrm{wt}(\mathbf{c}) = n-k+1$ if and only if there exists a $k$-subset $I_3\subseteq [n]$ such that $\sigma_{1}(\boldsymbol{\alpha}_{I_3})=0$ and $\eta\sigma_{k-\ell}(\boldsymbol{\alpha}_{I_3})=1$.

  \item [(3)] Since $\min\{\mathrm{wt}(\mathbf{c}) : \mathbf{c} \in T_3^{(\ell)} \setminus \{\mathbf{0}\}\} \geq n-k+1$, for a given $\mathbf{c} = (c_1, \ldots, c_{n+2}) \in T_3^{(\ell)}$, we have $\mathrm{wt}(\mathbf{c}) = n-k+1$ if and only if there exists a $k$-subset $I_4\subseteq [n]$ such that $\delta\sigma_{1}(\boldsymbol{\alpha}_{I_4})+\sigma_{2}(\boldsymbol{\alpha}_{I_4})=0$ and $\eta\sigma_{k-\ell}(\boldsymbol{\alpha}_{I_4})=1$.
        \end{itemize}

    In summary, for $0\leq\ell\leq k-3$, the minimum distance of the code $\operatorname{TRL}_{k,n+2}(\boldsymbol{\alpha},\eta,\delta,\ell)$ is $n-k+1$ if and only if
    \begin{enumerate}
         \item [(A)] no $k$-subset $I_1\subseteq [n]$ such that $\sigma_{1}(\boldsymbol{\alpha}_{I_1})=\sigma_{2}(\boldsymbol{\alpha}_{I_1})=0$ and $\eta\sigma_{k-\ell}(\boldsymbol{\alpha}_{I_1})=1$;
         \item [(B)] at least one of (i)-(iii) holds:
          \begin{enumerate}
          \item [(i)] There exists a $k-1$-subset $I_{2}\subseteq [n]$ such that
    \[
\sigma_{1}(\boldsymbol{\alpha}_{I_2})\notin\{\alpha_1,\cdots,\alpha_n\}\setminus\boldsymbol{\alpha}_{I_2},\ 
\sigma_{2}(\boldsymbol{\alpha}_{I_2})=\sigma_{1}^2(\boldsymbol{\alpha}_{I_2})\ \mbox{and}\  
\sigma_{k-\ell}(\boldsymbol{\alpha}_{I_2})=\sigma_{1}(\boldsymbol{\alpha}_{I_2})\sigma_{k-\ell-1}(\boldsymbol{\alpha}_{I_2})+\eta^{-1};
 \]
 \item [(ii)] there exists a $k$-subset $I_3\subseteq [n]$ such that $\sigma_{1}(\boldsymbol{\alpha}_{I_3})=0$ and $\eta\sigma_{k-\ell}(\boldsymbol{\alpha}_{I_3})=1$;
          \item [(iii)] there exists a $k$-subset $I_4\subseteq [n]$ such that $\delta\sigma_{1}(\boldsymbol{\alpha}_{I_4})+\sigma_{2}(\boldsymbol{\alpha}_{I_4})=0$ and $\eta\sigma_{k-\ell}(\boldsymbol{\alpha}_{I_4})=1$.
      \end{enumerate}
     \end{enumerate}

     \end{proof}

The following examples provide a code $\operatorname{TRL}_{k,n+2}(\boldsymbol{\alpha},\eta,\delta,\ell)$ with minimum distance $n-k+1$  for several representative choices of $\ell$.

\begin{example}
\begin{description}
    \item [\normalfont(1)]   Let $n=12,k=8,\boldsymbol{\alpha}=\{23, 3, 12, 7, 5, 8, 35, 19, 4, 36, 18, 22\}\subseteq\mathbb{F}_{37}$ and $\ell=k-2=6$. Choose $\eta=10,\delta=36$. On the one hand, there exists a $8$-subset $I= \{ 2, 3, 5, 6, 7, 8, 9, 10\}\subseteq [12]$ such that $\delta\eta\sigma_1(\boldsymbol{\alpha}_{I})=-1$ and $\eta\sigma_2(\boldsymbol{\alpha}_I)=1$.
     Thus, the minimum distance of the code $\operatorname{TRL}_{8,14}(\boldsymbol{\alpha},\eta,\delta,6)$ is $n+1-k=5$ from Lemma~\ref{Lem:d=n-k+1 and ell=k-2,k-1}.
On the other hand, from the given generator matrix, it can be directly verified  that this code is a linear code with parameters $[14,8,5]$ over $\mathbb{F}_{37}$ by using Magma.
\item [\normalfont(2)]   Let $n=12,k=8,\boldsymbol{\alpha}=\{ 15, 0, 17, 35, 30, 23, 4, 31, 13, 9, 21, 10\}\subseteq\mathbb{F}_{37}$ and $\ell=k-1=7$. Choose $\eta=6,\delta=13$. On the one hand, there exists a $8$-subset $I= \{2, 4, 5, 6, 7, 8,11,12\}\subseteq [12]$ such that $\sigma_{1}(\boldsymbol{\alpha}_{I})=\eta^{-1}$ and $\sigma_{2}(\boldsymbol{\alpha}_{I})=-\delta\eta^{-1}$. Thus, the minimum distance of the code $\operatorname{TRL}_{8,14}(\boldsymbol{\alpha},\eta,\delta,7)$ is $n+1-k=5$ from Lemma~\ref{Lem:d=n-k+1 and ell=k-2,k-1}.
On the other hand, from the given generator matrix, it can be directly verified  that this code is a linear code with parameters $[14,8,5]$ over $\mathbb{F}_{37}$ by using Magma.
\item [\normalfont(3)] Let $n=14,k=9,\boldsymbol{\alpha}=\{1, 4, 6, 36, 7, 9, 38, 10, 41, 17, 21, 23, 24, 26\}\subseteq\mathbb{F}_{43}$ and $\ell=5$. Choose $\eta=3,\delta=9$. On the one hand,  Magma verifies that  there does not exist a $9$-subset $I_1\subseteq [14]$ such that $\sigma_{1}(\boldsymbol{\alpha}_{I_1})=\sigma_{2}(\boldsymbol{\alpha}_{I_1})=0,\sigma_{4}(\boldsymbol{\alpha}_{I_1})=\eta^{-1}$ and there exists a $9$-subset $I_2= \{1, 2, 6, 8, 9, 10,12,13,14\}\subseteq [14]$ such that $\delta\cdot\sigma_{1}(\boldsymbol{\alpha}_{I_2})+\sigma_{2}(\boldsymbol{\alpha}_{I_2})=0$ and $\eta\cdot\sigma_{k-\ell}(\boldsymbol{\alpha}_{I_2})=1$. Thus, the minimum distance of the code $\operatorname{TRL}_{9,16}(\boldsymbol{\alpha},\eta,\delta,5)$ is $n+1-k=6$ from Lemma~\ref{Lem:d=n-k+1 and 0 le ell le k-3}.
On the other hand, from the given generator matrix, it can be directly verified  that this code is a linear code with parameters $[16,9,6]$ over $\mathbb{F}_{43}$ by using Magma.
 \end{description}
   
\end{example}

Next, we present necessary and sufficient conditions for the code $\operatorname{TRL}_{k,n+2}(\boldsymbol{\alpha},\eta,\delta,\ell)$  to be an MDS code.
The following lemma plays an important role in determining the necessary and sufficient conditions for these codes to be MDS codes.

\begin{lemma}[{\cite[Lemma 2.3]{yan2024mutually}}]\label{Lem:4.1}
	Let $m$ be a fixed positive integer and $$I_{m}=\left\{0,1,\cdots,m-1\right\}=\left\{t_{1},\cdots,t_{s}\right\}\bigcup\left\{r_{1},r_{2},\cdots,r_{s^{'}}\right\}$$ be any partition of $I_{m}$ with $m=s+s^{\prime},0=t_{1}<t_2<\cdots<t_{s}=m-1$ and $r_{1}<r_{2}<\cdots<r_{s^{'}}$.
	For any $\mathcal{S}=\left\{a_{1},a_{2},\cdots,a_{s}\right\}\subseteq \mathbb{F}_{q}$, we have the following determinant formula
	\begin{equation*}
	\det \begin{pmatrix}
	a_{1}^{t_{1}}&a_{2}^{t_{1}}&\cdots&a_{s}^{t_{1}}\\
	a_{1}^{t_{2}}&a_{2}^{t_{2}}&\cdots&a_{s}^{t_{2}}\\
	\vdots&\vdots&\vdots&\vdots\\
	a_{1}^{t_{s}}&a_{2}^{t_{s}}&\cdots&a_{s}^{t_{s}}\\
	\end{pmatrix}=\prod\limits_{1\leq i<j\leq s}(a_{j}-a_{i})\cdot\vartriangle,
	\end{equation*}
	where $\vartriangle$ denotes the determinant of the following matrix
	$$\begin{pmatrix}
	S_{s-r_{1}}(\mathcal{S})&S_{s-r_{2}}(\mathcal{S})&\cdots&S_{s-r_{s^{\prime}}}(\mathcal{S})\\
	S_{s-r_{1}+1}(\mathcal{S})&S_{s-r_{2}+1}(\mathcal{S})&\cdots&S_{s-r_{s^{\prime}}+1}(\mathcal{S})\\
	\vdots&\vdots&\vdots&\vdots\\
	S_{s-r_{1}+s^{\prime}-1}(\mathcal{S})&S_{s-r_{2}+s^{\prime}-1}(\mathcal{S})&\cdots&S_{s-r_{s^{\prime}}+s^{\prime}-1}(\mathcal{S})
	\end{pmatrix}.$$
	
\end{lemma}

\begin{theorem}\label{Thm: main MDS, 0 ell k-3}
      Let $n,k$ be integers satisfying $3\le k<n$. Let $\boldsymbol{\alpha}=\{\alpha_1,\ldots,\alpha_n\}\subseteq\mathbb F_q$ with $\alpha_i\neq\alpha_j$ for $i\neq j$ and $\eta,\delta\in\mathbb F_q^{*}$. For $0\leq\ell\leq k-3$, then the code $\operatorname{TRL}_{k,n+2}(\boldsymbol{\alpha},\eta,\delta,\ell)$ is an MDS code if and only if
      \begin{description}
		\item[\normalfont(i)] for any $k$-subset $\mathcal{L}\subseteq [n]$, we have $1+(-1)^{k-\ell-1}\eta S_{k-\ell}(\boldsymbol{\alpha}_{\mathcal{L}})\neq 0$.
		\item[\normalfont(ii)] for any $k-1$-subset $\mathcal{I}\subseteq [n]$, we have $1+(-1)^{k-\ell-2}\eta\left(S_{k-\ell-1}(\boldsymbol{\alpha}_{I})S_{1}(\boldsymbol{\alpha}_{I})-S_{k-\ell}(\boldsymbol{\alpha}_{I})\right)\neq 0$ and 
\[
\delta-S_{1}(\boldsymbol{\alpha}_{I})+(-1)^{k-\ell-2}\eta\left(\delta\left(
S_{k-\ell-1}(\boldsymbol{\alpha}_{I})S_{1}(\boldsymbol{\alpha}_{I})-S_{k-\ell}(\boldsymbol{\alpha}_{I})\right)-\left(S_{k-\ell-1}(\boldsymbol{\alpha}_{I})\cdot S_{2}(\boldsymbol{\alpha}_{I})-S_{k-\ell}(\boldsymbol{\alpha}_{I})\cdot S_{1}(\boldsymbol{\alpha}_{I})\right)\right)\neq 0.
\]
\item[\normalfont(iii)] for any $k-2$-subset $\mathcal{J}\subseteq [n]$, we have
$1+(-1)^{k-\ell-3}\eta\cdot\left|
    \begin{array}{ccc}
         S_{k-\ell-2}(\boldsymbol{\alpha}_{J})&1&0\\
         S_{k-\ell-1}(\boldsymbol{\alpha}_{J})&S_{1}(\boldsymbol{\alpha}_{J})&1\\
         S_{k-\ell}(\boldsymbol{\alpha}_{J})&S_{2}(\boldsymbol{\alpha}_{J})&S_{1}(\boldsymbol{\alpha}_{J})
    \end{array}
    \right|
    \neq 0$.
	 \end{description}


\end{theorem}

\begin{proof}
    It is well-known that an $[n+2,k]$-linear code is MDS if and only if any $k$ columns of its generator matrix are linearly independent. Therefore, it is enough to prove that any $k\times k$ submatrix of $G_{k,n+2}^{\operatorname{TRL}}(\boldsymbol{\alpha},\eta,\delta,\ell)$ is invertible. For $0\leq\ell\leq k-3$, it suffices to show that for any subsets $\mathcal{L}=\left\{i_{1},\cdots,i_{k}\right\},\mathcal{I}=\left\{i_{1},\cdots,i_{k-1}\right\}$ and $\mathcal{J}=\left\{i_{1},\cdots,i_{k-2}\right\}$ of $[n]$, the matrices
    \begin{equation*}
 	A=\small{\begin{pmatrix}
		1&\cdots&1\\
		\alpha_{i_{1}}&\cdots&\alpha_{i_{k}}\\
		\vdots&\vdots&\vdots\\
		\alpha_{i_1}^{\ell-1}&\cdots&\alpha_{i_{k}}^{\ell-1}\\
\alpha_{i_1}^{\ell}+\eta\alpha_{i_1}^{k}&\cdots&\alpha_{i_{k}}^{\ell}+\eta\alpha_{i_{k}}^k\\
		\alpha_{i_1}^{\ell+1}&\cdots&\alpha_{i_{k}}^{\ell+1}\\
		\vdots&\vdots&\vdots\\
		\alpha_{i_1}^{k-1}&\cdots&\alpha_{i_{k}}^{k-1}\\
		\end{pmatrix}},
	C_{t}=\small{\begin{pmatrix}
		1&\cdots&1&0\\
		\alpha_{i_{1}}&\cdots&\alpha_{i_{k-1}}&0\\
		\vdots&\vdots&\vdots&\vdots\\
		\alpha_{i_1}^{\ell-1}&\cdots&\alpha_{i_{k-1}}^{\ell-1}&0\\
        \alpha_{i_1}^{\ell}+\eta\alpha_{i_1}^{k}&\cdots&\alpha_{i_{k-1}}^{\ell}+\eta\alpha_{i_{k-1}}^k&0\\
		\alpha_{i_1}^{\ell+1}&\cdots&\alpha_{i_{k-1}}^{\ell+1}&0\\
		\vdots&\vdots&\vdots&\vdots\\
		\alpha_{i_1}^{k-3}&\cdots&\alpha_{i_{k-1}}^{k-3}&0\\
        \alpha_{i_1}^{k-2}&\cdots&\alpha_{i_{k-1}}^{k-2}&c_{t1}\\
        \alpha_{i_1}^{k-1}&\cdots&\alpha_{i_{k-1}}^{k-1}&c_{t2}\\
		\end{pmatrix}},t=1,2
	\end{equation*}
	and 
	\begin{equation*}
     D=\small{\begin{pmatrix}
     	1&\cdots&1&0&0\\
     	\alpha_{i_{1}}&\cdots&\alpha_{i_{k-2}}&0&0\\
     	\vdots&\vdots&\vdots&\vdots&\vdots\\
     	\alpha_{i_{1}}^{\ell-1}&\cdots&\alpha_{i_{k-2}}^{\ell-1}&0&0\\
     	\alpha_{i_1}^{\ell}+\eta\alpha_{i_1}^{k}&\cdots&\alpha_{i_{k-2}}^{\ell}+\eta\alpha_{i_{k-2}}^k&0&0\\
     	\alpha_{i_{1}}^{\ell+1}&\cdots&\alpha_{i_{k-2}}^{\ell+1}&0&0\\
        \vdots&\vdots&\vdots&\vdots&\vdots\\
        \alpha_{i_{1}}^{k-3}&\cdots&\alpha_{i_{k-2}}^{k-3}&0&0\\
        \alpha_{i_{1}}^{k-2}&\cdots&\alpha_{i_{k-2}}^{k-2}&0&1\\
        \alpha_{i_{1}}^{k-1}&\cdots&\alpha_{i_{k-2}}^{k-1}&1&\delta\\
     	\end{pmatrix}}
	\end{equation*}
	are all invertible, where $(c_{11},c_{12})^T=(0,1)$ and  $(c_{21},c_{22})^T=(1,\delta)$.
    From Lemma~\ref{Lem:4.1}, we know that
	\begin{equation*}
	\begin{aligned}
	\det(A)&=V(\boldsymbol{\alpha}_{\mathcal{L}})+(-1)^{k-\ell-1}\eta S_{k-\ell}(\boldsymbol{\alpha}_{\mathcal{L}})\cdot V(\boldsymbol{\alpha}_{\mathcal{L}})\\
	&=V(\boldsymbol{\alpha}_{\mathcal{L}})\left(1+(-1)^{k-\ell-1}\eta S_{k-\ell}(\boldsymbol{\alpha}_{\mathcal{L}})\right).
	\end{aligned}
	\end{equation*} 
    Thus, $\det(A)\neq 0$ if and only if $1+(-1)^{k-\ell-1}\eta S_{k-\ell}(\boldsymbol{\alpha}_{\mathcal{L}})\neq 0$. We know that  
    \begin{equation*}
        \begin{aligned}
           \det(C_{1})&=V(\boldsymbol{\alpha}_{I})+(-1)^{k-\ell-2}\eta\cdot V(\boldsymbol{\alpha}_{I})\cdot\left|
    \begin{array}{cc}
         S_{k-\ell-1}(\boldsymbol{\alpha}_{I})&1\\
    S_{k-\ell}(\boldsymbol{\alpha}_{I})&S_{1}(\boldsymbol{\alpha}_{I})
    \end{array}\right| \\
    &=V(\boldsymbol{\alpha}_{I})\cdot\left(1+(-1)^{k-\ell-2}\eta\cdot\left(S_{k-\ell-1}(\boldsymbol{\alpha}_{I})S_{1}(\boldsymbol{\alpha}_{I})-S_{k-\ell}(\boldsymbol{\alpha}_{I})\right)\right).
        \end{aligned}
    \end{equation*}
Thus, $\det(C_{1})\neq 0$ if and only if $1+(-1)^{k-\ell-2}\eta\cdot\left( S_{k-\ell-1}(\boldsymbol{\alpha}_{I})S_{1}(\boldsymbol{\alpha}_{I})-S_{k-\ell}(\boldsymbol{\alpha}_{I})\right)\neq 0$.  We know that
    \begin{equation*}
        \begin{aligned}
            \det(C_2)&=\delta\cdot\left|
            \begin{array}{ccc}
                1&\cdots&1\\
		\alpha_{i_{1}}&\cdots&\alpha_{i_{k-1}}\\
		\vdots&\vdots&\vdots\\
		\alpha_{i_1}^{\ell-1}&\cdots&\alpha_{i_{k-1}}^{\ell-1}\\
        \alpha_{i_1}^{\ell}+\eta\alpha_{i_1}^{k}&\cdots&\alpha_{i_{k-1}}^{\ell}+\eta\alpha_{i_{k-1}}^k\\
		\alpha_{i_1}^{\ell+1}&\cdots&\alpha_{i_{k-1}}^{\ell+1}\\
		\vdots&\vdots&\vdots\\
		\alpha_{i_1}^{k-3}&\cdots&\alpha_{i_{k-1}}^{k-3}\\
        \alpha_{i_1}^{k-2}&\cdots&\alpha_{i_{k-1}}^{k-2}
            \end{array}\right|-\left|
            \begin{array}{ccc}
                1&\cdots&1\\
		\alpha_{i_{1}}&\cdots&\alpha_{i_{k-1}}\\
		\vdots&\vdots&\vdots\\
		\alpha_{i_1}^{\ell-1}&\cdots&\alpha_{i_{k-1}}^{\ell-1}\\
        \alpha_{i_1}^{\ell}+\eta\alpha_{i_1}^{k}&\cdots&\alpha_{i_{k-1}}^{\ell}+\eta\alpha_{i_{k-1}}^k\\
		\alpha_{i_1}^{\ell+1}&\cdots&\alpha_{i_{k-1}}^{\ell+1}\\
		\vdots&\vdots&\vdots\\
		\alpha_{i_1}^{k-3}&\cdots&\alpha_{i_{k-1}}^{k-3}\\
        \alpha_{i_1}^{k-1}&\cdots&\alpha_{i_{k-1}}^{k-1}
            \end{array}\right|\\
            &=\delta\cdot V(\boldsymbol{\alpha}_{I})+(-1)^{k-\ell-2}\delta\eta\cdot V(\boldsymbol{\alpha}_{I})\cdot(S_{k-\ell-1}(\boldsymbol{\alpha}_{I})\cdot S_{1}(\boldsymbol{\alpha}_{I})-S_{k-\ell}(\boldsymbol{\alpha}_{I}))\\
            &-V(\boldsymbol{\alpha}_{I})\cdot S_{1}(\boldsymbol{\alpha}_{I})-(-1)^{k-\ell-2}\eta\cdot V(\boldsymbol{\alpha}_{I})\cdot\left(S_{k-\ell-1}(\boldsymbol{\alpha}_{I})\cdot S_{2}(\boldsymbol{\alpha}_{I})-S_{k-\ell}(\boldsymbol{\alpha}_{I})\cdot S_{1}(\boldsymbol{\alpha}_{I})\right)
        \end{aligned}
    \end{equation*}
Thus, $\det(C_{2})\neq 0$ if and only if
\[
\delta-S_{1}(\boldsymbol{\alpha}_{I})+(-1)^{k-\ell-2}\eta\left(\delta\left(
S_{k-\ell-1}(\boldsymbol{\alpha}_{I})S_{1}(\boldsymbol{\alpha}_{I})-S_{k-\ell}(\boldsymbol{\alpha}_{I})\right)-\left(S_{k-\ell-1}(\boldsymbol{\alpha}_{I})\cdot S_{2}(\boldsymbol{\alpha}_{I})-S_{k-\ell}(\boldsymbol{\alpha}_{I})\cdot S_{1}(\boldsymbol{\alpha}_{I})\right)\right)\neq 0.
\]
We know that
    $$
    \det(D)=-V(\boldsymbol{\alpha}_{\mathcal{J}})-(-1)^{k-\ell-3}\eta\cdot V(\boldsymbol{\alpha}_{J})\cdot\left|
    \begin{array}{ccc}
         S_{k-\ell-2}(\boldsymbol{\alpha}_{J})&1&0\\
         S_{k-\ell-1}(\boldsymbol{\alpha}_{J})&S_{1}(\boldsymbol{\alpha}_{J})&1\\
         S_{k-\ell}(\boldsymbol{\alpha}_{J})&S_{2}(\boldsymbol{\alpha}_{J})&S_{1}(\boldsymbol{\alpha}_{J})
    \end{array}
    \right|
    $$
	Thus, $\det(D)\neq 0$ if and only if $1+(-1)^{k-\ell-3}\eta\cdot\left|
    \begin{array}{ccc}
         S_{k-\ell-2}(\boldsymbol{\alpha}_{J})&1&0\\
         S_{k-\ell-1}(\boldsymbol{\alpha}_{J})&S_{1}(\boldsymbol{\alpha}_{J})&1\\
         S_{k-\ell}(\boldsymbol{\alpha}_{J})&S_{2}(\boldsymbol{\alpha}_{J})&S_{1}(\boldsymbol{\alpha}_{J})
    \end{array}
    \right|
    \neq 0$.

Thus, for $0\leq\ell\leq k-3$, the code $\operatorname{TRL}_{k,n+2}(\boldsymbol{\alpha},\eta,\delta,\ell)$ is an MDS code if and only if
      \begin{enumerate}
		\item[(i)] for any $k$-subset $\mathcal{L}\subseteq [n]$, we have $1+(-1)^{k-\ell-1}\eta S_{k-\ell}(\boldsymbol{\alpha}_{\mathcal{L}})\neq 0$.
		\item[(ii)] for any $k-1$-subset $\mathcal{I}\subseteq [n]$, we have $1+(-1)^{k-\ell-2}\eta\cdot\left(S_{k-\ell-1}(\boldsymbol{\alpha}_{I})S_{1}(\boldsymbol{\alpha}_{I})-S_{k-\ell}(\boldsymbol{\alpha}_{I})\right)\neq 0$ and 
\[
\delta-S_{1}(\boldsymbol{\alpha}_{I})+(-1)^{k-\ell-2}\eta\left(\delta\left(
S_{k-\ell-1}(\boldsymbol{\alpha}_{I})S_{1}(\boldsymbol{\alpha}_{I})-S_{k-\ell}(\boldsymbol{\alpha}_{I})\right)-\left(S_{k-\ell-1}(\boldsymbol{\alpha}_{I})\cdot S_{2}(\boldsymbol{\alpha}_{I})-S_{k-\ell}(\boldsymbol{\alpha}_{I})\cdot S_{1}(\boldsymbol{\alpha}_{I})\right)\right)\neq 0.
\]
\item[(iii)] for any $k-2$-subset $\mathcal{J}\subseteq [n]$, we have
$1+(-1)^{k-\ell-3}\eta\cdot\left|
    \begin{array}{ccc}
         S_{k-\ell-2}(\boldsymbol{\alpha}_{J})&1&0\\
         S_{k-\ell-1}(\boldsymbol{\alpha}_{J})&S_{1}(\boldsymbol{\alpha}_{J})&1\\
         S_{k-\ell}(\boldsymbol{\alpha}_{J})&S_{2}(\boldsymbol{\alpha}_{J})&S_{1}(\boldsymbol{\alpha}_{J})
    \end{array}
    \right|
    \neq 0$.
	\end{enumerate}
\end{proof}
Employing the same proof method, we establish the necessary and sufficient conditions under which $\operatorname{TRL}_{k,n+2}(\boldsymbol{\alpha},\eta,\delta,\ell)$ is an MDS code for $k-2\leq\ell\leq k-1$.

\begin{theorem}\label{Thm: main MDS, ell k-2}
      Let $n,k$ be integers satisfying $3\le k<n$. Let $\boldsymbol{\alpha}=\{\alpha_1,\ldots,\alpha_n\}\subseteq\mathbb F_q$ with $\alpha_i\neq\alpha_j$ for $i\neq j$ and $\eta,\delta\in\mathbb F_q^{*}$. Then 
the code $\operatorname{TRL}_{k,n+2}(\boldsymbol{\alpha},\eta,\delta,k-2)$ is an MDS code if and only if
      \begin{description}
		\item[\normalfont(i)] for any $k$-subset $\mathcal{L}\subseteq [n]$, we have $\eta\cdot S_{2}(\boldsymbol{\alpha}_{\mathcal L})\neq 1$.
		\item[\normalfont(ii)] for any $k-1$-subset $\mathcal{I}\subseteq [n]$, we have $1+\eta(S_{1}^2(\boldsymbol{\alpha}_{\mathcal I})-S_{2}(\boldsymbol \alpha_{\mathcal I}))\neq 0$ and 
        \[
        \delta\left(1+\eta(S_{1}^2(\boldsymbol \alpha_{\mathcal I})- S_{2}(\boldsymbol \alpha_{\mathcal I}))\right)-S_{1}(\boldsymbol \alpha_{\mathcal I})\neq 0.
        \]
     \end{description}
\end{theorem}

\begin{theorem}\label{Thm: main MDS, ell k-1}
      Let $n,k$ be integers satisfying $3\le k<n$. Let $\boldsymbol{\alpha}=\{\alpha_1,\ldots,\alpha_n\}\subseteq\mathbb F_q$ with $\alpha_i\neq\alpha_j$ for $i\neq j$ and $\eta,\delta\in\mathbb F_q^{*}$.  Then 
      the code $\operatorname{TRL}_{k,n+2}(\boldsymbol{\alpha},\eta,\delta,k-1)$ is an MDS code if and only if
      \begin{description}
		\item[\normalfont(i)] for any $k$-subset $\mathcal{L}\subseteq [n]$, we have $\eta\cdot S_{1}(\boldsymbol{\alpha}_{\mathcal L})\neq -1$.
		\item[\normalfont(ii)] for any $k-1$-subset $\mathcal{I}\subseteq [n]$, we have $\delta-S_{1}(\boldsymbol \alpha_{\mathcal I})-\eta\left(S_{1}^2(\boldsymbol{\alpha}_{\mathcal{I}})-S_{2}(\boldsymbol{\alpha}_{\mathcal{I}})\right)\neq 0$.
        \end{description}
\end{theorem}

Finally, we proceed to show the dimension of the Schur square code of $\operatorname{TRL}_{k,n+2}(\boldsymbol{\alpha},\eta,\delta,\ell)$, which will in turn show that this code  is a non-RS code.

\begin{theorem}\label{Thm: non-RS}
      Let $n,k,\ell$ be integers satisfying $5\le k\leq\frac{n-1}{2}$ and $0\le \ell\le k-1$. Let $\boldsymbol{\alpha}=\{\alpha_1,\ldots,\alpha_n\}\subseteq\mathbb F_q$ with $\alpha_i\neq\alpha_j$ for $i\neq j$.  Let $\eta,\delta\in\mathbb{F}_{q}^{*}$, then
      the dimension of the Schur square of code $\mathcal{C}=\operatorname{TRL}_{k,n+2}(\boldsymbol{\alpha},\eta,\delta,\ell)$  satisfies 
      $$\dim(\mathcal C^2)=\left\{
      \begin{array}{cc}
         2k+1,  &\mbox{if}\ k-2\leq\ell\leq k-1  \\
          2k+2, &\mbox{if}\ \ell=0\\
          2k+3,&\mbox{if}\ 1\leq\ell\leq k-3
      \end{array}\right..$$
      Thus, the code $\operatorname{TRL}_{k,n+2}(\boldsymbol{\alpha},\eta,\delta,\ell)$ is a non-RS code distinct from the $\operatorname{RL}$ code $\operatorname{RL}_{k,n+2}(\boldsymbol{\alpha},\delta)$.
\end{theorem}
\begin{proof}
    Let $\mathcal{C}:=\operatorname{TRL}_{k,n+2}(\boldsymbol{\alpha},\eta,\delta,\ell)$ and $G=G_{k,n+2}^{\operatorname{TRL}}(\boldsymbol{\alpha},\eta,\delta,\ell)$, we show the dimension of the Schur square of $\mathcal{C}$. Let $\boldsymbol{g}_{i}$ be the $i+1$-th row of the matrix $G$ and $\boldsymbol{g}_{i,j}=\boldsymbol{g}_{i}\ast \boldsymbol{g}_{j}$ for any $0\leq i,j\leq k-1$.  We divide our discussion into four cases:
    \begin{itemize}
        \item[(i)] If $\ell=0$, then $\mathcal{C}^2$ is generated by the following matrix
        $$\begin{pmatrix}
\alpha_1^2 & \alpha_2^2 & \cdots & \alpha_n^2&0&0 \\
\alpha_1^3 & \alpha_2^3 & \cdots & \alpha_n^3&0&0 \\
\vdots & \vdots & \ddots & \vdots&\vdots&\vdots \\
\alpha_1^{2k-6} & \alpha_2^{2k-6} & \cdots & \alpha_n^{2k-6}&0&0 \\
(1+\eta\alpha_{1}^k)^2&(1+\eta\alpha_{1}^k)^2&\cdots&(1+\eta\alpha_{1}^k)^2&0&0\\
\alpha_1(1+\eta\alpha_{1}^k)&\alpha_2(1+\eta\alpha_{1}^k)&\cdots&\alpha_{n}(1+\eta\alpha_{1}^k)&0&0\\
\vdots & \vdots & \ddots & \vdots&\vdots&\vdots \\
\alpha_1^{k-1}(1+\eta\alpha_{1}^k)&\alpha_2^{k-1}(1+\eta\alpha_{1}^k)&\cdots&\alpha_{n}^{k-1}(1+\eta\alpha_{1}^k)&0&0\\
\alpha_{1}^{2k-4}&\alpha_{2}^{2k-4}&\cdots&\alpha_{n}^{2k-4}&0&1\\
\alpha_{1}^{2k-3}&\alpha_{2}^{2k-3}&\cdots&\alpha_{n}^{2k-3}&0&\delta\\
\alpha_{1}^{2k-2}&\alpha_{2}^{2k-2}&\cdots&\alpha_{n}^{2k-2}&1&\delta^2\\
\end{pmatrix},$$
which is  row equivalent to
$$\begin{pmatrix}
\alpha_1 & \alpha_2 & \cdots & \alpha_n&0&0 \\
\alpha_1^2 & \alpha_2^2 & \cdots & \alpha_n^2&0&0 \\
\vdots & \vdots & \ddots & \vdots&\vdots&\vdots \\
\alpha_1^{2k-3} & \alpha_2^{2k-3} & \cdots & \alpha_n^{2k-3}&0&0 \\
\alpha_1^{2k-2} & \alpha_2^{2k-2} & \cdots & \alpha_n^{2k-2}&0&0 \\
\alpha_1^{2k-1} & \alpha_2^{2k-1} & \cdots & \alpha_n^{2k-1}&0&0 \\
1+\eta^2\alpha_1^{2k} & 1+\eta^2\alpha_{2k} & \cdots & 1+\eta^2\alpha_n^{2k} &1&0\\
0&0&\cdots&0&0&1\\
0&0&\cdots&0&1&0\\
\end{pmatrix},$$
whose rank is $2k+2$ because of $2k+2\leq n+2$.
\item [(ii)] If $1\leq\ell\leq k-3$, $\mathcal{C}^2$ is generated by a matrix equivalent to the following form:
\[
\begin{pmatrix}
1 & 1 & \cdots & 1 & 0 & 0\\
\alpha_1 & \alpha_2 & \cdots & \alpha_n & 0 & 0\\
\vdots & \vdots & \ddots & \vdots & \vdots & \vdots  \\
\alpha_1^{2k-3} & \alpha_2^{2k-3} & \cdots & \alpha_n^{2k-3} & 0 & 0\\
\alpha_1^{2k-2} & \alpha_2^{2k-2} & \cdots & \alpha_n^{2k-2} & 0 & 0 \\
\alpha_1^{2k-1} & \alpha_2^{2k-1} & \cdots & \alpha_n^{2k-1} & 0 & 0\\
\alpha_{1}^{2k}&\alpha_{2}^{2k}&\cdots&\alpha_{n}^{2k}&0&0\\
0&0&\cdots&0&0&1\\
0&0&\cdots&0&1&0\\
\end{pmatrix},
\]
whose rank is $2k+3$ because of $2k+3\leq n+2$.
\item [(iii)] If $\ell=k-2$, $\mathcal{C}^2$ is generated by a matrix equivalent to the following form:
\[
\begin{pmatrix}
1 & 1 & \cdots & 1 & 0 & 0\\
\alpha_1 & \alpha_2 & \cdots & \alpha_n & 0 & 0\\
\vdots & \vdots & \ddots & \vdots & \vdots & \vdots  \\
\alpha_{1}^{2k-4}&\alpha_{2}^{2k-4}&\cdots&\alpha_{n}^{2k-4}&0&0\\
\alpha_1^{2k-3} & \alpha_2^{2k-3} & \cdots & \alpha_n^{2k-3} & 0 & 0\\
\alpha_1^{2k-2} & \alpha_2^{2k-2} & \cdots & \alpha_n^{2k-2} & 1 & \delta^2 \\
\eta\alpha_1^{2k-1} & \eta\alpha_2^{2k-1} & \cdots & \eta\alpha_n^{2k-1} & 0 & \delta\\
\eta^2\alpha_{1}^{2k}+2\eta\alpha_1^{2k-2}&\eta^2\alpha_{2}^{2k}+2\eta\alpha_2^{2k-2}&\cdots&\eta^2\alpha_{n}^{2k}+2\eta\alpha_n^{2k-2}&0&1\\
\end{pmatrix},
\]
whose rank is $2k+1$ because of $2k<2k+1\leq n$.

\item [(iv)] If $\ell=k-1$, $\mathcal{C}^2$ is generated by a matrix equivalent to the following form:
\[
\begin{pmatrix}
1 & 1 & \cdots & 1 & 0 & 0\\
\alpha_1 & \alpha_2 & \cdots & \alpha_n & 0 & 0\\
\vdots & \vdots & \ddots & \vdots & \vdots & \vdots  \\
\alpha_{1}^{2k-4}&\alpha_{2}^{2k-4}&\cdots&\alpha_{n}^{2k-4}&0&0\\
\eta\alpha_1^{2k-3} & \eta\alpha_2^{2k-3} & \cdots & \eta\alpha_n^{2k-3} & 0 & 0\\
\eta\alpha_1^{2k-2} & \eta\alpha_2^{2k-2} & \cdots & \eta\alpha_n^{2k-2} & 0 & 0 \\
2\eta\alpha_1^{2k-1}+\eta^2\alpha_{1}^{2k} & 2\eta\alpha_2^{2k-1}+\eta^2\alpha_{2}^{2k} & \cdots & 2\eta\alpha_n^{2k-1}+\eta^2\alpha_{n}^{2k} & 1 & 0\\
0&0&\cdots&0&0&1
\end{pmatrix},
\]
whose rank is $2k+1$ because of $2k+1<n+2$.
\end{itemize}

In summary, the dimension of the Schur square of code $\mathcal{C}$  satisfies 
      $$\dim(\mathcal C^2)=\left\{
      \begin{array}{cc}
         2k+1,  &\mbox{if}\ k-2\leq\ell\leq k-1  \\
          2k+2, &\mbox{if}\ \ell=0\\
          2k+3,&\mbox{if}\ 1\leq\ell\leq k-3
      \end{array}\right..$$
By Proposition~\ref{Prop:schur of RS}, the dimension of the Schur square of RS codes is $2k-1$, and since the dimension of the Schur square of $\operatorname{RL}$ codes $\operatorname{RL}_{k,n+2}(\boldsymbol{\alpha},\delta)$ is $2k$, the code $\mathcal{C}$ is a non-RS code inequivalent to the corresponding $\operatorname{RL}$ code.
\end{proof}

Combining the MDS criteria of Theorems~\ref{Thm: main MDS, 0 ell k-3}---\ref{Thm: main MDS, ell k-1} with the Schur square distinguisher in Theorem~\ref{Thm: non-RS}, we can construct a family of non-RS MDS codes.
\begin{proposition}
    Let $p$ be a prime, $q_1 = p^{m_1}$, $q_2 = p^{m_2}$, where $m_1 \mid m_2$ and $m_2 \geq 3m_1$. Let $\boldsymbol{\alpha}=\{\alpha_{1},\cdots,\alpha_{n}\}\subseteq\mathbb{F}_{q_1}$, $\delta\in\mathbb{F}_{q_2}^{*}\setminus\mathbb{F}_{q_1}$, where $\alpha_1, \cdots, \alpha_n$ are distinct in $\mathbb{F}_{q_1}$. Let $H_1=\left\{a_1+a_2\delta:a_1,a_2\in\mathbb{F}_{q_1}\right\},H_2=\left\{a_1+a_2\delta^{-1}:a_1,a_2\in\mathbb{F}_{q_1}\right\}$ and $\eta\in\mathbb{F}_{q_2}^{*}\setminus \left(H_1\cup H_2\right)$.
    For $k-2\leq\ell\leq k-1$,
   if $5\leq k\leq\frac{n-1}{2}$, then  the code $\operatorname{TRL}_{k,n+2}(\boldsymbol{\alpha},\eta,\delta,\ell)$ over $\mathbb F_{q_2}$ is a non-RS MDS code.
\end{proposition}

\begin{proof}
Let
$
\mathcal{C}_{\ell}=\operatorname{TRL}_{k,n+2}(\boldsymbol{\alpha},\eta,\delta,\ell)$ for $k-2\leq\ell\leq k-1$.
By Theorem~\ref{Thm: non-RS}, $\mathcal{C}_{\ell}$ is a non-RS code. Therefore, it suffices to show that $\mathcal{C}_{\ell}$ is an MDS code. We distinguish the following two cases for $\ell=k-1$:
\begin{enumerate}
    \item[(i)] For $k$-subsets $L\subseteq [n]$, since $\eta^{-1}\in\mathbb{F}_{q_2}^{*}\setminus\mathbb{F}_{q_1}$ and $-S_1(\boldsymbol{\alpha}_{L})\in\mathbb{F}_{q_1}$, we have $\eta\cdot S_1(\boldsymbol{\alpha}_{L})\neq -1$.
    \item [(ii)] For any $k-1$-subset $I\subseteq [n]$, since $\delta\in \mathbb F^*_{q_2}\setminus\mathbb F_{q_1}$ and $S_1(\boldsymbol{\alpha}_{I})\in\mathbb{F}_{q_1}$, we have $\delta-S_1(\boldsymbol{\alpha}_{I})\neq 0$. If $S_{1}^2(\boldsymbol{\alpha}_{I})-S_2(\boldsymbol{\alpha}_{I})=0$, then
    $\delta-S_{1}(\boldsymbol \alpha_{ I})-\eta\left(S_{1}^2(\boldsymbol{\alpha}_{{I}})-S_{2}(\boldsymbol{\alpha}_{{I}})\right)=\delta-S_1(\boldsymbol{\alpha}_{I})\neq 0$. If $S_{1}^2(\boldsymbol{\alpha}_{I})-S_2(\boldsymbol{\alpha}_{I})\neq 0$, since $\left(\delta-S_{1}(\boldsymbol \alpha_{ I})\right)\cdot \left(S_{1}^2(\boldsymbol{\alpha}_{{I}})-S_{2}(\boldsymbol{\alpha}_{{I}})\right)^{-1}\in H_1$ and $\eta\in\mathbb{F}_{q_2}^{*}\setminus \left(H_1\cup H_2\right)$, we have 
    \[
    \delta-S_{1}(\boldsymbol \alpha_{ I})-\eta\left(S_{1}^2(\boldsymbol{\alpha}_{{I}})-S_{2}(\boldsymbol{\alpha}_{{I}})\right)\neq 0.
    \]
    \end{enumerate}
    
    Combining the above two cases with Theorem~\ref{Thm: main MDS, ell k-1}, we know that $\mathcal{C}_{k-1}$ is an MDS code. Therefore, $\mathcal{C}_{k-1}$ is a non-RS MDS code.

    To prove that $\mathcal{C}_{k-2}$ is an MDS code, we divide the proof into the following two cases.
    \begin{enumerate}
        \item [(i)] For $k$-subsets $L\subseteq [n]$, since $\eta^{-1}\in\mathbb{F}_{q_2}^{*}\setminus\mathbb{F}_{q_1}$ and $S_2(\boldsymbol{\alpha}_{L})\in\mathbb{F}_{q_1}$, we have $\eta\cdot S_2(\boldsymbol{\alpha}_{L})\neq 1$.
        \item [(ii)] For any $k-1$-subset $I\subseteq [n]$, since $\eta\in\mathbb F^{*}_{q_2}\setminus\mathbb{F}_{q_1}$ and $S_{1}^2(\boldsymbol{\alpha}_{I})-S_{2}(\boldsymbol{\alpha}_{I})\in\mathbb F_{q_1}$, we have $1+\eta\left(S_{1}^2(\boldsymbol{\alpha}_{I})-S_{2}(\boldsymbol{\alpha}_{I})\right)\neq 0$. Next, we prove 
        \[
        \delta(1+\eta(S_{1}^2(\boldsymbol{\alpha}_{I})-S_2(\boldsymbol{\alpha}_{I})))-S_1(\boldsymbol{\alpha}_{I})\neq 0.
        \]
        If $S_{1}^2(\boldsymbol{\alpha}_{I})-S_2(\boldsymbol{\alpha}_{I})=0$, then 
  \[
        \delta(1+\eta(S_{1}^2(\boldsymbol{\alpha}_{I})-S_2(\boldsymbol{\alpha}_{I})))-S_1(\boldsymbol{\alpha}_{I})=\delta-S_{1}(\boldsymbol{\alpha}_{I})\neq 0.
        \]

        If $S_{1}^2(\boldsymbol{\alpha}_{I})-S_2(\boldsymbol{\alpha}_{I})\neq 0$, since $\left(\delta^{-1}\cdot S_{1}(\boldsymbol{\alpha}_{I})-1\right)\cdot\left(S_{1}^2(\boldsymbol{\alpha}_{I})-S_2(\boldsymbol{\alpha}_{I})\right)^{-1}\in H_1\cup H_2$ and $\eta\in\mathbb{F}_{q_2}^{*}\setminus (H_1\cup H_2)$, we have
        \[
        \eta\neq \left(\delta^{-1}\cdot S_{1}(\boldsymbol{\alpha}_{I})-1\right)\cdot\left(S_{1}^2(\boldsymbol{\alpha}_{I})-S_2(\boldsymbol{\alpha}_{I})\right)^{-1},
        \]
        i.e. $\delta\left(1+\eta(S_{1}^2(\boldsymbol \alpha_{\mathcal I})- S_{2}(\boldsymbol \alpha_{\mathcal I}))\right)-S_{1}(\boldsymbol \alpha_{\mathcal I})\neq 0$.
\end{enumerate}

   Combining the above two cases with Theorem~\ref{Thm: main MDS, ell k-2}, we know that $\mathcal{C}_{k-2}$ is an MDS code. Therefore, $\mathcal{C}_{k-2}$ is a non-RS MDS code.
\end{proof}

\begin{proposition}
Let $\mathbb F_{q_1}\subseteq\mathbb F_{q_2}\subseteq \mathbb F_q$ be a chain of subfields of $\mathbb F_q$. Let $\boldsymbol{\alpha}=\{\alpha_{1},\cdots,\alpha_{n}\}\subseteq\mathbb{F}_{q_1}$, $\delta\in\mathbb{F}_{q_2}^{*}\setminus\mathbb{F}_{q_1}$ and $\eta\in\mathbb F_q^{*}\setminus\mathbb F_{q_2}$, where $\alpha_1, \cdots, \alpha_n$ are distinct in $\mathbb{F}_{q_1}$. 
    For $0\leq\ell\leq k-3$,
   if $5\leq k\leq\frac{n-1}{2}$, then  the code $\operatorname{TRL}_{k,n+2}(\boldsymbol{\alpha},\eta,\delta,\ell)$ over $\mathbb F_q$ is a non-RS MDS code.
\end{proposition}

\begin{proof}
Let
$
\mathcal{C}_{\ell}=\operatorname{TRL}_{k,n+2}(\boldsymbol{\alpha},\eta,\delta,\ell)$ for $0\leq\ell\leq k-3$.
By Theorem~\ref{Thm: non-RS}, $\mathcal{C}_{\ell}$ is a non-RS code. Therefore, it suffices to show that $\mathcal{C}_{\ell}$ is an MDS code. We distinguish the following three cases for $0\leq\ell\leq k-3$:
\begin{enumerate}
    \item[(i)] For a $k$-subset $L\subseteq [n]$, if
    $S_{k-\ell}(\boldsymbol{\alpha}_L)=0$, then $1+(-1)^{k-\ell-1}\eta\, S_{k-\ell}(\boldsymbol{\alpha}_L)\ne 0$. If $S_{k-\ell}(\boldsymbol{\alpha}_L)\ne 0$, since
    $\eta^{-1}\in \mathbb F_{q}^{*}\setminus\mathbb F_{q_2}$ and 
    $(-1)^{k-\ell}S_{k-\ell}(\boldsymbol{\alpha}_L)\in\mathbb F_{q_1}\subseteq \mathbb F_{q_2}$,
    we have
    $1+(-1)^{k-\ell-1}\eta\,S_{k-\ell}(\boldsymbol{\alpha}_L)\ne 0$. Hence, for every $k$-subset $L\subseteq [n]$, we always have
    \[
    1+(-1)^{k-\ell-1}\eta\,S_{k-\ell}(\boldsymbol{\alpha}_L)\ne 0.
    \]

    \item[(ii)] For the $k-1$ subsets $I\subseteq [n]$, since
    $\eta\in\mathbb F_{q}^{*}\setminus\mathbb F_{q_2}$ and $S_{k-\ell-1}(\boldsymbol{\alpha}_I)S_1(\boldsymbol{\alpha}_I)-S_{k-\ell}(\boldsymbol{\alpha}_I)\in\mathbb F_{q_1}$,
    it follows that
    \[
    1+(-1)^{k-\ell-2}\eta\bigl(S_{k-\ell-1}(\boldsymbol{\alpha}_I)S_1(\boldsymbol{\alpha}_I)-S_{k-\ell}(\boldsymbol{\alpha}_I)\bigr)\ne 0.
    \]
 Since
   $ S_1(\boldsymbol{\alpha}_I)\in\mathbb F_{q_1}$ and $ \delta\in \mathbb F_{q_2}^*\setminus\mathbb F_{q_1}$, we get
    $\delta-S_1(\boldsymbol{\alpha}_I)\ne 0$. Since $\delta-S_1(\boldsymbol{\alpha}_I)\ne 0,\eta\in\mathbb F_q^{*}\setminus\mathbb F_{q_2}$ and
   \[
   \delta\left(
S_{k-\ell-1}(\boldsymbol{\alpha}_{I})S_{1}(\boldsymbol{\alpha}_{I})-S_{k-\ell}(\boldsymbol{\alpha}_{I})\right)-\left(S_{k-\ell-1}(\boldsymbol{\alpha}_{I})\cdot S_{2}(\boldsymbol{\alpha}_{I})-S_{k-\ell}(\boldsymbol{\alpha}_{I})\cdot S_{1}(\boldsymbol{\alpha}_{I})\right)\in\mathbb F_{q_2},
   \]
    we obtain
    \[
\delta-S_{1}(\boldsymbol{\alpha}_{I})+(-1)^{k-\ell-2}\eta\left(\delta\left(
S_{k-\ell-1}(\boldsymbol{\alpha}_{I})S_{1}(\boldsymbol{\alpha}_{I})-S_{k-\ell}(\boldsymbol{\alpha}_{I})\right)-\left(S_{k-\ell-1}(\boldsymbol{\alpha}_{I})\cdot S_{2}(\boldsymbol{\alpha}_{I})-S_{k-\ell}(\boldsymbol{\alpha}_{I})\cdot S_{1}(\boldsymbol{\alpha}_{I})\right)\right)\neq 0.
\]

    \item[(iii)] For $k-2$ subsets $J\subseteq [n]$, since
    $\eta\in\mathbb F_{q}^{*}\setminus\mathbb F_{q_2}$
    and 
    \[
    \left|
    \begin{array}{ccc}
         S_{k-\ell-2}(\boldsymbol{\alpha}_{J})&1&0\\
         S_{k-\ell-1}(\boldsymbol{\alpha}_{J})&S_{1}(\boldsymbol{\alpha}_{J})&1\\
         S_{k-\ell}(\boldsymbol{\alpha}_{J})&S_{2}(\boldsymbol{\alpha}_{J})&S_{1}(\boldsymbol{\alpha}_{J})
    \end{array}
    \right|\in\mathbb{F}_{q_1}\subseteq \mathbb F_{q_2},
    \]
    we have
    $1+(-1)^{k-\ell-3}\eta\cdot\left|
    \begin{array}{ccc}
         S_{k-\ell-2}(\boldsymbol{\alpha}_{J})&1&0\\
         S_{k-\ell-1}(\boldsymbol{\alpha}_{J})&S_{1}(\boldsymbol{\alpha}_{J})&1\\
         S_{k-\ell}(\boldsymbol{\alpha}_{J})&S_{2}(\boldsymbol{\alpha}_{J})&S_{1}(\boldsymbol{\alpha}_{J})
    \end{array}
    \right|
    \neq 0$.
\end{enumerate}

 For $0\leq\ell\leq k-3$, combining the above three cases with Theorem~\ref{Thm: main MDS, 0 ell k-3}, we know that $\mathcal{C}_{\ell}$ is an MDS code. Therefore, $\mathcal{C}_{\ell}$ is a non-RS MDS code.
\end{proof}

\begin{example}
   Let $n=11,k=5,\boldsymbol{\alpha}=\{\alpha_1,\alpha_2,\cdots,\alpha_{11}\}\subseteq\mathbb{F}_{11}$, where $\alpha_{i}=i-1$ for all $1\leq i\leq 11$. Let $\delta$ be a primitive element of $\mathbb{F}_{11^3},H_1=\left\{a_1+a_2\delta:a_1,a_2\in\mathbb{F}_{11}\right\}$ and $H_2=\left\{a_1+a_2\delta^{-1}:a_1,a_2\in\mathbb{F}_{11}\right\}$. Let $\eta=\delta+\delta^{-1}$, because $\delta$ is a primitive element of $\mathbb{F}_{11^3}$, we know that $\eta\in \mathbb F_{11^3}\setminus \left(H_1\cup H_2\right)$. 
   \begin{description}
\item [\normalfont(1)] If $\ell=4$, then the code $\operatorname{TRL}_{5,13}(\boldsymbol{\alpha},\eta,\delta,4)$ is an MDS code with parameters $[13,5,9]$ over $\mathbb{F}_{11^3}$ by Magma.  Moreover, the Schur square code of $\operatorname{TRL}_{5,13}(\boldsymbol{\alpha},\eta,\delta,4)$ has dimension $11\neq 9$. Thus, the codes $\operatorname{TRL}_{5,13}(\boldsymbol{\alpha},\eta,\delta,4)$ are non-RS MDS codes.
\item [\normalfont(2)] If $\ell=3$, then the code $\operatorname{TRL}_{5,13}(\boldsymbol{\alpha},\eta,\delta,3)$ is an MDS code with parameters $[13,5,9]$ over $\mathbb{F}_{11^3}$ by Magma.  Moreover, the Schur square code of $\operatorname{TRL}_{5,13}(\boldsymbol{\alpha},\eta,\delta,3)$ has dimension $11\neq 9$. Thus, the codes $\operatorname{TRL}_{5,13}(\boldsymbol{\alpha},\eta,\delta,3)$ are non-RS MDS codes.
    \end{description}
   \end{example}
   \begin{example}
   Let $n=11,k=5,\boldsymbol{\alpha}=\{\alpha_1,\alpha_2,\cdots,\alpha_{11}\}\subseteq\mathbb{F}_{11}$, where $\alpha_{i}=i-1$ for all $1\leq i\leq 11$. Let $\delta$ be a primitive element of $\mathbb{F}_{11^2}$ and $\eta$ be a primitive element of $\mathbb{F}_{11^4}$, then $\delta\in\mathbb F_{11^2}^{*}\setminus\mathbb F_{11}$ and $\eta\in\mathbb F_{11^4}^{*}\setminus\mathbb F_{11^2}$.
   \begin{description}
   \item [\normalfont(1)] If $\ell=0$, then the code $\operatorname{TRL}_{5,13}(\boldsymbol{\alpha},\eta,\delta,0)$ is an MDS code with parameters $[13,5,9]$ over $\mathbb{F}_{11^4}$ by Magma.  Moreover, the Schur square code of $\operatorname{TRL}_{5,13}(\boldsymbol{\alpha},\eta,\delta,0)$ has dimension $12\neq 9$. Thus, the codes $\operatorname{TRL}_{5,13}(\boldsymbol{\alpha},\eta,\delta,0)$ are non-RS MDS codes.
\item [\normalfont(2)] If $\ell=1$, then the code $\operatorname{TRL}_{5,13}(\boldsymbol{\alpha},\eta,\delta,1)$ is an MDS code with parameters $[13,5,9]$ over $\mathbb{F}_{11^4}$ by Magma.  Moreover, the Schur square code of $\operatorname{TRL}_{5,13}(\boldsymbol{\alpha},\eta,\delta,1)$ has dimension $13\neq 9$. Thus, the codes $\operatorname{TRL}_{5,13}(\boldsymbol{\alpha},\eta,\delta,1)$ are non-RS MDS codes.
\item [\normalfont(3)] If $\ell=2$, then the code $\operatorname{TRL}_{5,13}(\boldsymbol{\alpha},\eta,\delta,2)$ is an MDS code with parameters $[13,5,9]$ over $\mathbb{F}_{11^4}$ by Magma.  Moreover, the Schur square code of $\operatorname{TRL}_{5,13}(\boldsymbol{\alpha},\eta,\delta,2)$ has dimension $13\neq 9$. Thus, the codes $\operatorname{TRL}_{5,13}(\boldsymbol{\alpha},\eta,\delta,2)$ are non-RS MDS codes.
    \end{description}
\end{example}

\section{The NMDS Characterization  of code \texorpdfstring{$\operatorname{TRL}_{k,n+2}(\boldsymbol{\alpha},\eta,\delta,\ell)$ }{TRL(k,n+2,alpha,eta,delta,ell)}}\label{sec4}
In this section, we give the necessary and sufficient conditions for the code $\operatorname{TRL}_{k,n+2}(\boldsymbol{\alpha},\eta,\delta,\ell)$ to be an NMDS code for $k-2\leq\ell\leq k-1$.

The following lemma is a key result about the NMDS property of linear codes, which is useful to determine the NMDS property for $\operatorname{TRL}$ codes.
\begin{lemma}\cite{dodunekov1994near}\label{Lem:equ NMDS}
    Let $\mathcal{C}$ be a linear code over $\mathbb F_{q}$ with parameters $[n,k]$. Then $\mathcal{C}$ is an NMDS code if and only if the following conditions hold:
    \begin{description}
        \item [\normalfont(i)] Any $k-1$ columns of a generator matrix of $\mathcal{C}$ are linearly independent;
        \item [\normalfont(ii)] There exist $k$ columns of a generator matrix of $\mathcal{C}$ that are linearly dependent;
        \item [\normalfont(iii)]  Any $k+1$ columns of a generator matrix of $\mathcal{C}$ have rank $k$.
     \end{description}
\end{lemma}

\begin{theorem}\label{Thm: NMDS,ell=k-1}
      Let $n,k$ be integers satisfying $3\le k<n$. Let $\boldsymbol{\alpha}=\{\alpha_1,\ldots,\alpha_n\}\subseteq\mathbb F_q^*$ with $\alpha_i\neq\alpha_j$ for $i\neq j$ and $\eta,\delta\in\mathbb F_q^{*}$.  Then the code $\operatorname{TRL}_{k,n+2}(\boldsymbol{\alpha},\eta,\delta,k-1)$ is an NMDS code if and only if 
      \begin{description}
          \item [\normalfont(A)] no $k$-subset $J\subseteq [n]$ satisfies $S_1(\boldsymbol{\alpha}_J)=-\eta^{-1},S_{2}(\boldsymbol{\alpha}_{J})=-\delta\eta^{-1}$;
           and 
           \item [\normalfont(B)] at least one of the following conditions holds:
           \begin{description} 
		\item[\normalfont(i)] there exists a $k$-subset $\mathcal{L}\subseteq [n]$ such that $\eta\cdot S_{1}(\boldsymbol{\alpha}_{\mathcal L})= -1$.
		\item[\normalfont(ii)] there exists a $k-1$-subset $\mathcal{I}\subseteq [n]$ such that $\delta-S_{1}(\boldsymbol \alpha_{\mathcal I})-\eta\left(S_{1}^2(\boldsymbol{\alpha}_{\mathcal{I}})-S_{2}(\boldsymbol{\alpha}_{\mathcal{I}})\right)= 0$
        \end{description} 
        \end{description}
\end{theorem}
\begin{proof}
    Let $\mathcal{C}=\operatorname{TRL}_{k,n+2}(\boldsymbol{\alpha},\eta,\delta,k-1)$ and $G=G_{k,n+2}^{\operatorname{TRL}}(\boldsymbol{\alpha},\eta,\delta,k-1)$. 
Given a column index set $I\subseteq [n+2]$ and a row index set $J\subseteq [k]$, let $G(I,J)$ denote the $|J|\times |I|$ submatrix of $G$ formed by the column set $I$ and the row set $J$. 
In particular, $G_I$ denotes the submatrix $G(I,[k])$. 
We first prove that for any $k-1$-subset $I\subseteq [n+2]$, the matrix $G_I$ has full rank. 
We consider the following three cases.
\begin{itemize}
    \item [(1.1)] For any $k-1$-subset $I\subseteq [n]$, since $\det(G(I,[k-1]))=V(\boldsymbol{\alpha}_I)\neq 0$, the matrix $G_I$ has full rank.
\item [(1.2)] For any $k-2$-subset $I\subseteq [n]$, let $I'=I\cup\{n+1\}$. Since 
\[
\det(G(I',[2,k]))=\prod_{i\in I}\alpha_i\cdot V(\boldsymbol{\alpha}_I)\neq 0,
\]
the matrix $G_{I'}$ has full rank. Let $I''=I\cup\{n+2\}$. Since 
\[
\det(G(I'',[k-1]))=V(\boldsymbol{\alpha}_I)\neq 0,
\]
the matrix $G_{I''}$ has full rank.
\item [(1.3)] For any $k-3$-subset $J\subseteq [n]$, let $J'=J\cup\{n+1,n+2\}$. Since 
\[
\det(G(J',[2,k]))=-\prod_{i\in J}\alpha_i\cdot V(\boldsymbol{\alpha}_{J})\neq 0,
\]
the matrix $G_{J'}$ has full rank. 
\end{itemize}
Therefore, for any $k-1$-subset $I\subseteq [n+2]$, the matrix $G_I$ has full rank.

By Theorem~\ref{Thm: main MDS, ell k-1}, the generator matrix $G$ of the code $\mathcal{C}$ has $k$ linearly dependent columns if and only if $\mathcal{C}$ is not an MDS code, if and only if one of the following two conditions holds:
\begin{itemize}
    \item[(2.1)] there exists a $k$-subset $L\subseteq [n]$ such that $\eta\cdot S_1(\boldsymbol{\alpha}_L)=-1$;
    \item[(2.2)] there exists a $k-1$-subset $I\subseteq [n]$ such that $\delta-S_1(\boldsymbol{\alpha}_I)-\eta\left(S_1^2(\boldsymbol{\alpha}_I)-S_2(\boldsymbol{\alpha}_I)\right)=0$.
\end{itemize}

Finally, we prove that any $k+1$ columns of generator matrix $G$ have rank $k$, if and only if there does not exist a $k$-subset $J\subseteq [n]$ such that
\[
S_1(\boldsymbol{\alpha}_J)=-\eta^{-1},\quad\mbox{and}\quad S_2(\boldsymbol{\alpha}_J)=\eta^{-1}\big(S_1(\boldsymbol{\alpha}_J)-\delta\big)-S_1^2(\boldsymbol{\alpha}_J).
\]

For any $J=\{j_1,j_2,\ldots,j_{k+1}\}\subseteq [n+2]$ and a subset $I\subseteq [k+1]$, where $1\leq j_1<\cdots<j_{k+1}\leq n+2$, let $J_I=J\setminus\{j_i:i\in I\}$. In particular, for $i\in [k+1]$ and $J_i=J\setminus\{j_i\}$. We consider the following four cases.
\begin{itemize}
    \item [(3.1)]  If $j_k<j_{k+1}\leq n$, that is, $J\subseteq [n]$. It is known that for  $1\leq i\leq k+1$, we have
\[
\det(G(J_i,[k]))=V(\boldsymbol{\alpha}_{J_i})+\eta V(\boldsymbol{\alpha}_{J_i})\cdot S_1(\boldsymbol{\alpha}_{J_i})=V(\boldsymbol{\alpha}_{J_i})\cdot\left(1+\eta\cdot S_1(\boldsymbol{\alpha}_{J_i})\right).
\]
Since for $i_1\neq i_2\in J$,
\[
\big(1+\eta S_1(\boldsymbol{\alpha}_{J_{i_1}})\big)-\big(1+\eta S_1(\boldsymbol{\alpha}_{J_{i_2}})\big)=\eta(\boldsymbol{\alpha}_{j_{i_2}}-\boldsymbol{\alpha}_{j_{i_1}})\neq 0,
\]
there does not exist $i_1\neq i_2$ such that $\det(G(J_{i_1},[k]))=\det(G(J_{i_2},[k]))=0$. Therefore, there exists a $k$-subset $J'\subseteq J$ such that $\det(G(J',[k]))\neq 0$, that is, the rank of the matrix $G_J$ is $k$.

\item [(3.2)]  If $(j_k,j_{k+1})=(n+1,n+2)$, let $J_1^{'}=J_1\setminus\{n+1,n+2\}$. Then
\[
\det(G(J_1,[k]))=-V(\boldsymbol{\alpha}_{J_1^{'}})\neq 0.
\]
Therefore, the rank of the matrix $G_J$ is $k$.

\item[(3.3)]  If $j_k\leq n$ and $j_{k+1}=n+1$, let $J_1^{''}=J_{1}\setminus\{n+1\}$. Then
\[
\det(G(J_{1},[k]))=V(\boldsymbol{\alpha}_{J_1^{''}})\neq 0.
\]
Therefore, the rank of the matrix $G_J$ is $k$.

\item [(3.4)] If $j_k\leq n$ and $j_{k+1}=n+2$, combining Cases (3.1), (3.2), and (3.3), we will prove that the rank of $G_J$ is less than $k$ if and only if
\[
S_1(\boldsymbol{\alpha}_{J_{k+1}})=-\eta^{-1}\quad\text{and}\quad S_2(\boldsymbol{\alpha}_{J_{k+1}})= \eta^{-1}\big(S_1(\boldsymbol{\alpha}_{J_{k+1}})-\delta\big)+S_1^2(\boldsymbol{\alpha}_{J_{k+1}}).
\]

For $1\leq i\leq k$, let $J_{i,k+1}=J\setminus\{j_i,j_{k+1}\}$. Let $S_1=S_1(\boldsymbol{\alpha}_{J_{k+1}})$ and $S_2=S_2(\boldsymbol{\alpha}_{J_{k+1}})$. Then $$S_1(\boldsymbol{\alpha}_{J_{i,k+1}})=S_1-\alpha_{j_i}$$
and
\begin{align*}
S_2(\boldsymbol{\alpha}_{J_{i,k+1}})&=S_2(\boldsymbol{\alpha}_{J_{k+1}})-\alpha_{j_i}\cdot S_1(\boldsymbol{\alpha}_{J_{i,k+1}})\\
&=S_2-\alpha_{j_i}\cdot (S_1-\alpha_{j_i})=S_2-\alpha_{j_i}\cdot S_1+\alpha_{j_i}^2.
\end{align*}

Thus,
\begin{align*}
\det(G(J_{i},[k]))&=\delta\cdot V(\boldsymbol{\alpha}_{J_{i,k+1}})-\Big(V(\boldsymbol{\alpha}_{J_{i,k+1}})\cdot S_1(\boldsymbol{\alpha}_{J_{i,k+1}})+\eta V(\boldsymbol{\alpha}_{J_{i,k+1}})\cdot\begin{vmatrix} S_1(\boldsymbol{\alpha}_{J_{i,k+1}}) & 1 \\ S_2(\boldsymbol{\alpha}_{J_{i,k+1}}) & S_1(\boldsymbol{\alpha}_{J_{i,k+1}}) \end{vmatrix}\Big)\\
&=V(\boldsymbol{\alpha}_{J_{i,k+1}})\Big(\delta-S_1(\boldsymbol{\alpha}_{J_{i,k+1}})-\eta\big(S_1^2(\boldsymbol{\alpha}_{J_{i,k+1}})-S_2(\boldsymbol{\alpha}_{J_{i,k+1}})\big)\Big)\\
&=V(\boldsymbol{\alpha}_{J_{i,k+1}})\Big(\delta-S_1+\eta(S_2-S_1^2)+\alpha_{j_i}(\eta S_1+1)\Big)
\end{align*}
for all $1\leq i\leq k$ and
\[
\det(G(J_{k+1},[k]))=V(\boldsymbol{\alpha}_{J_{k+1}})+\eta V(\boldsymbol{\alpha}_{J_{k+1}})\cdot S_1(\boldsymbol{\alpha}_{J_{k+1}})=V(\boldsymbol{\alpha}_{J_{k+1}})(1+\eta S_1).
\]

Since $V(\boldsymbol{\alpha}_{J_{k+1}})\neq 0$ and $V(\boldsymbol{\alpha}_{J_{i,k+1}})\neq 0$ for all $1\leq i\leq k$, the rank of the matrix $G_J$ is less than $k$ if and only if $\det(G(J_i,[k]))=0$ for all $1\leq i\leq k$ and $\det(J_{k+1},[k])=0$, if and only if
\[
S_1=-\eta^{-1}\quad\text{and}\quad \delta-S_1+\eta(S_2-S_1^2)=0,
\]
if and only if 
\[S_1=-\eta^{-1}\quad\mbox{and}\quad S_2=-\delta\eta^{-1}.\]
\end{itemize}

 By Lemma~\ref{Lem:equ NMDS} and combining the above cases, we conclude that the code $\mathcal C$ is an NMDS code if and only if
 \begin{itemize}
          \item [(A)] no $k$-subset $J\subseteq [n]$ satisfies $S_1(\boldsymbol{\alpha}_J)=-\eta^{-1},S_{2}(\boldsymbol{\alpha}_{J})=-\delta\eta^{-1}$;
           and 
           \item [(B)] at least one of (i) or (ii) holds:
           \begin{itemize}
		\item[(i)] there exists a $k$-subset $\mathcal{L}\subseteq [n]$ such that $\eta\cdot S_{1}(\boldsymbol{\alpha}_{\mathcal L})= -1$.
		\item[(ii)] there exists a $k-1$-subset $\mathcal{I}\subseteq [n]$ such that $\delta-S_{1}(\boldsymbol \alpha_{\mathcal I})-\eta\left(S_{1}^2(\boldsymbol{\alpha}_{\mathcal{I}})-S_{2}(\boldsymbol{\alpha}_{\mathcal{I}})\right)= 0$
        \end{itemize}
      \end{itemize}

\end{proof}

\begin{theorem}\label{Thm: NMDS,ell=k-2}
      Let $n$ and $k$ be integers satisfying $3\leq k<n$. Let $\boldsymbol{\alpha}=\{\alpha_1,\ldots,\alpha_n\}\subseteq \mathbb{F}_q^*$ with $\alpha_i\neq\alpha_j$ for $i\neq j$, and let $\eta,\delta\in\mathbb{F}_q^*$. Then the code $\operatorname{TRL}_{k,n+2}(\boldsymbol{\alpha},\eta,\delta,k-2)$ is an NMDS code if and only if the following conditions hold:
\begin{description}
    \item [\normalfont(A)]For every $(k-1)$-subset $I\subseteq[n]$, $\bigl(S_1(\boldsymbol{\alpha}_I),S_2(\boldsymbol{\alpha}_I)\bigr)\in\mathbb{F}_q^2\setminus\{(0,\eta^{-1})\}$.
    \item [\normalfont(B)]For every $k$-subset $J\subseteq[n]$, $\bigl(S_1(\boldsymbol{\alpha}_J),S_2(\boldsymbol{\alpha}_J)\bigr)\in\mathbb{F}_q^2\setminus\{(0,\eta^{-1}),(\eta^{-1}\delta^{-1},\eta^{-1})\}$.
    \item [\normalfont(C)] At least one of the following conditions holds:
    \begin{description}
        \item [(i)] There exists a $k$-subset $L\subseteq[n]$ such that $\eta S_2(\boldsymbol{\alpha}_L)=1$.
        \item [(ii)] There exists a $(k-1)$-subset $I\subseteq[n]$ such that $1+\eta\bigl(S_1^2(\boldsymbol{\alpha}_I)-S_2(\boldsymbol{\alpha}_I)\bigr)\in\{0,\delta^{-1}S_1(\boldsymbol{\alpha}_I)\}$.
     \end{description} \end{description}
\end{theorem}
\begin{proof}
    Let $\mathcal{C}=\operatorname{TRL}_{k,n+2}(\boldsymbol{\alpha},\eta,\delta,k-2)$ and $G=G_{k,n+2}^{\operatorname{TRL}}(\boldsymbol{\alpha},\eta,\delta,k-2)$. 
Given a column index set $I\subseteq [n+2]$ and a row index set $J\subseteq [k]$, let $G(I,J)$ denote the $|J|\times |I|$ submatrix of $G$ formed by the set of columns $I$ and the set of rows $J$. 
In particular, $G_I$ denotes the submatrix $G(I,[k])$. 
We first prove that any $k-1$ columns of the generator matrix $G$ are linearly independent if and only if for any $k-1$-subset ${I}\subseteq [n]$, we have $S_1(\boldsymbol{\alpha}_{I})\neq 0$ or $S_{2}(\boldsymbol{\alpha}_{I})\neq\eta^{-1}$.
We consider the following three cases.
\begin{itemize}
    \item [(1.1)] For any $k-1$-subset $I\subseteq [n]$, we have \[\det(G(I,[k-1]))=V(\boldsymbol{\alpha}_{I})\left(1+\eta\cdot (S_{1}^2(\boldsymbol{\alpha}_{I})-S_2(\boldsymbol{\alpha}_{I})\right)\ ,\ 
    \det(G(I,[k]\setminus\{k-1\}))=V(\boldsymbol{\alpha}_{I})\cdot S_{1}(\boldsymbol{\alpha}_{I})
    \]
    and
    \[
    \det(G(I,[k]\setminus\{i\})=V(\boldsymbol{\alpha}_{I})\cdot\left(S_{k-i}(\boldsymbol{\alpha}_{I})-\eta\cdot S_2(\boldsymbol{\alpha}_{I})\cdot S_{k-i}(\boldsymbol{\alpha}_{I})+\eta\cdot S_1(\boldsymbol{\alpha}_{I})\cdot S_{k-i+1}(\boldsymbol{\alpha}_{I})\right)
    \]
    for all $1\leq i\leq k-2$, where $S_{k}(\boldsymbol{\alpha}_{I})=0$. We prove that the matrix $G_I$ has full rank if and only if $S_{1}(\boldsymbol{\alpha}_{I})\neq 0$ or $S_2(\boldsymbol{\alpha}_{I})\neq\eta^{-1}$.

    If $S_{1}(\boldsymbol{\alpha}_{I})\neq 0$ or $S_2(\boldsymbol{\alpha}_{I})\neq\eta^{-1}$, since $S_{k-1}(\boldsymbol{\alpha}_{I})=\prod\limits_{i\in I}\alpha_{i}\neq 0$ and $V(\boldsymbol{\alpha}_{I})\neq 0$, then $\det(G(I,[k]\setminus\{k-1\}))\neq 0$ or $\det(G(I,[k]\setminus\{1\}))\neq 0$. Thus, the matrix $G_I$ has full rank.

    If  the matrix $G_I$ has full rank, then there exists $1\leq i\leq k$ such that $\det(G(I,[k]\setminus\{i\})\neq 0$. If $\det(G(I,[k]\setminus\{k-1\})\neq 0$ or $\det(G(I,[k]\setminus\{1\})\neq 0$, then $S_{1}(\boldsymbol{\alpha}_{I})\neq 0$ or $S_2(\boldsymbol{\alpha}_{I})\neq\eta^{-1}$. If $\det(G(I,[k]\setminus\{1\})=\det(G(I,[k]\setminus\{k-1\})=0$, then $S_{1}(\boldsymbol{\alpha}_{I})=0$ and $S_2(\boldsymbol{\alpha}_{I})=\eta^{-1}$. Thus,
$$\det(G(I,[k-1]))=V(\boldsymbol{\alpha}_{I})\left(1+\eta\cdot (S_{1}^2(\boldsymbol{\alpha}_{I})-S_2(\boldsymbol{\alpha}_{I})\right)=0
    $$
    and
    \begin{equation}
        \begin{aligned}
    \det(G(I,[k]\setminus\{i\})&=V(\boldsymbol{\alpha}_{I})\cdot\left(S_{k-i}(\boldsymbol{\alpha}_{I})-\eta\cdot S_2(\boldsymbol{\alpha}_{I})\cdot S_{k-i}(\boldsymbol{\alpha}_{I})+\eta\cdot S_1(\boldsymbol{\alpha}_{I})\cdot S_{k-i+1}(\boldsymbol{\alpha}_{I})\right)\\
    &=V(\boldsymbol{\alpha}_{I})\cdot S_{k-i}(\boldsymbol{\alpha}_{I})\cdot\left(1-\eta\cdot S_{2}(\boldsymbol{\alpha}_{I})\right)=0
        \end{aligned}
    \end{equation}
    for all $i\in [k]\setminus\{1,k-1,k\}$,
     which contradicts the fact that the matrix $G_I$ has full rank.  Therefore, the matrix $G_I$ has full rank, then $\det(G(I,[k]\setminus\{1\})\neq 0$ or $\det(G(I,[k]\setminus\{k-1\})\neq 0$, which implies $S_{1}(\boldsymbol{\alpha}_{I})\neq 0$ or $S_2(\boldsymbol{\alpha}_{I})\neq\eta^{-1}$.
\item [(1.2)] For any $k-2$-subset $I\subseteq [n]$, let $I'=I\cup\{n+1\}$. Since 
\[
\det(G(I',[k]\setminus\{k-1\}))=V(\boldsymbol{\alpha}_I)\neq 0,
\]
the matrix $G_{I'}$ has full rank. Let $I''=I\cup\{n+2\}$. Since 
\[
\det(G(I'',[k-1]))=V(\boldsymbol{\alpha}_I)\neq 0,
\]
the matrix $G_{I''}$ has full rank.
\item [(1.3)] For any $k-3$-subset $J\subseteq [n]$, let $J'=J\cup\{n+1,n+2\}$. Since 
\[
\det(G(J',[2,k]))=-V(\boldsymbol{\alpha}_{J})\cdot\prod_{i\in J}\alpha_i\neq 0,
\]
the matrix $G_{J'}$ has full rank.
\end{itemize}

Therefore, any $k-1$ columns of the generator matrix $G$ are linearly independent if and only if for any $k-1$-subset ${I}\subseteq [n]$, we have $S_1(\boldsymbol{\alpha}_{I})\neq 0$ or $S_{2}(\boldsymbol{\alpha}_{I})\neq\eta^{-1}$.


By Theorem~\ref{Thm: main MDS, ell k-2}, the generator matrix $G$ of the code $\mathcal C$ has $k$ linearly dependent columns if and only if $C$ is not an MDS code, if and only if one of the following two conditions holds:
      \begin{enumerate}
      \item[(2.1)] there exists a $k$-subset ${L}\subseteq [n]$ such that $\eta\cdot S_{2}(\boldsymbol{\alpha}_{ L})= 1$.
		\item[(2.2)] there exists a $k-1$-subset ${I}\subseteq [n]$ such that $1+\eta(S_{1}^2(\boldsymbol{\alpha}_{ I})-S_{2}(\boldsymbol \alpha_{ I}))= 0$ or
        \[
        \delta\left(1+\eta(S_{1}^2(\boldsymbol \alpha_{ I})- S_{2}(\boldsymbol \alpha_{ I}))\right)-S_{1}(\boldsymbol \alpha_{ I})= 0.
        \]
        \end{enumerate}

Finally, we prove that any $k+1$ columns of the generator matrix $G$ have rank $k$, if and only if for any $k$-subset $J\subseteq [n]$, we have $[S_1(\boldsymbol{\alpha}_{J})\neq 0\ \mbox{or}\ S_2(\boldsymbol{\alpha}_{J})\neq \eta^{-1}]$ and $[S_1(\boldsymbol{\alpha}_{J})\neq \eta^{-1}\delta^{-1}\ \mbox{or}\ S_2(\boldsymbol{\alpha}_{J})\neq\eta^{-1}]$.

For any $J=\{j_1,j_2,\ldots,j_{k+1}\}\subseteq [n+2]$ and a subset $I\subseteq [k+1]$, where $1\leq j_1<\cdots<j_{k+1}\leq n+2$, let $J_I=J\setminus\{j_i:i\in I\}$. In particular, for $i^{'}\neq i^{''} \in [k+1]$, let $J_i=J\setminus\{j_i\}$ and $J_{i^{'},i^{''}}=J\setminus\{j_{i^{'}},j_{i^{''}}\}$. Let $S_1=S_1(\boldsymbol{\alpha}_{J_{k+1}})$ and $S_2=S_2(\boldsymbol{\alpha}_{J_{k+1}})$. Then
\[
S_1(\boldsymbol{\alpha}_{J_{i,k+1}})=S_1-\alpha_{j_{i}}\quad\mbox{and}\quad
S_2(\boldsymbol{\alpha}_{J_{i,k+1}})=S_2-\alpha_{j_{i}}\cdot S_1+\alpha_{j_{i}}^2.
\]

Next, we consider the following four cases.
\begin{itemize}
    \item [(3.1)]  If $j_k<j_{k+1}\leq n$, that is, $J\subseteq [n]$. It is known that for  $1\leq i\leq k+1$, we have
\[
\det(G(J_i,[k]))=V(\boldsymbol{\alpha}_{J_i})-\eta V(\boldsymbol{\alpha}_{J_i})\cdot S_2(\boldsymbol{\alpha}_{J_i})=V(\boldsymbol{\alpha}_{J_i})\cdot\left(1-\eta\cdot S_2(\boldsymbol{\alpha}_{J_i})\right).
\]
Firstly, there exists a $k-1$-subset $J^{'}\subseteq J$ such that $S_{1}(\boldsymbol{\alpha}_{J^{'}})\neq 0$. For simplicity, let $J^{\prime}=J\setminus\{j_1,j_2\}$. Let $J^{(1)}=J^{\prime}\cup\{j_1\}$ and $J^{(2)}=J^{\prime}\cup\{j_2\}$. If $\det(G(J^{(1)},[k]))=\det(G(J^{(2)},[k]))=0$, then 
\[
1-\eta\cdot S_2(\boldsymbol{\alpha}_{J^{(1)}})=1-\eta\cdot S_2(\boldsymbol{\alpha}_{J^{(2)}}),
\]
that is, $\eta(\alpha_{j_{1}}-\alpha_{j_{2}})\cdot S_{1}(\boldsymbol{\alpha}_{J^{'}})=0$, which contradicts $\eta,\alpha_{j_{1}}-\alpha_{j_{2}},S_{1}(\boldsymbol{\alpha}_{J^{'}})\in\mathbb{F}_{q}^{*}$.
Thus, there exists $1\leq t\leq 2$ such that $\det(G(J^{(t)},[k])\neq 0$. In other words, the rank of matrix $G_{J}$ is $k$.

\item [(3.2)]  If $(j_k,j_{k+1})=(n+1,n+2)$, let $J_1^{'}=J_1\setminus\{n+1,n+2\}$. Then
\[
\det(G(J_1,[k]))=-V(\boldsymbol{\alpha}_{J_1^{'}})\neq 0.
\]
Therefore, the rank of the matrix $G_J$ is $k$.

\item[(3.3)]  If $j_k\leq n$ and $j_{k+1}=n+1$, then 
\begin{equation*}
    \begin{aligned}
        \det(G(J_{i},[k]))&=V(\boldsymbol{\alpha}_{J_{i,k+1}})\left(1+\eta\cdot \left(S_{1}^2(\boldsymbol{\alpha}_{J_{i,k+1}})-S_{2}(\boldsymbol{\alpha}_{J_{i,k+1}})\right)\right)\\
        &=V(\boldsymbol{\alpha}_{J_{i,k+1}})\cdot \left(1+\eta\left((S_1-\alpha_{j_i}\right)^2-(S_2-\alpha_{j_{i}}\cdot S_1+\alpha_{j_{i}}^2))\right)\\
        &=V(\boldsymbol{\alpha}_{J_{i,k+1}})\cdot (1+\eta(S_1^2-\alpha_{j_i}\cdot S_1-S_2))\\
        &=V(\boldsymbol{\alpha}_{J_{i,k+1}})(1+\eta(S_1^2-S_2)-\alpha_{j_i}\cdot\eta\cdot S_1)
    \end{aligned}
\end{equation*}
and
\begin{equation*}
    \det(G(J_{k+1},[k]))=V(\boldsymbol{\alpha}_{J_{k+1}})\left(1-\eta\cdot S_2\right).
\end{equation*}
We now prove that the rank of the matrix $G_J$ is $k$ if and only if $S_1\neq 0$ or $S_2\neq\eta^{-1}$. If $S_1\neq 0$, then there exists $1\leq i\leq k$ such that $1+\eta(S_1^2-S_2)-\alpha_{j_i}\cdot\eta\cdot S_1\neq 0$. Thus, there exists $1\leq i\leq k$ such that $\det(G(J_i,[k]))\neq 0$. If $S_2\neq\eta^{-1}$, then $\det(G(J_{k+1},[k]))\neq 0$. Combining these two cases, we know that when $S_1\neq 0$ or $S_2\neq \eta^{-1}$,  the rank of the matrix $G_J$ is $k$. If the rank of the matrix $G_J$ is $k$, then there exists $1\leq i\leq k+1$ such that $\det(G(J_{i},[k]))\neq 0$. If $\det(G(J_{k+1},[k]))\neq 0$, then $S_2\neq \eta^{-1}$. If $\det(G(J_{k+1},[k]))=0$ and there exists $1\leq i\leq k$ such that $\det(G(J_{i},[k]))\neq 0$. Then $S_2=\eta^{-1}$ and $1+\eta(S_1^2-S_2)-\alpha_{j_i}\cdot\eta\cdot S_1\neq 0$. Thus,$\eta\cdot S_1(S_1-\alpha_{j_i})\neq 0$. So $S_1\neq 0$. In summary, we have proved that the rank of the matrix $G_J$ is $k$ if and only if $S_1\neq 0$ or $S_2\neq\eta^{-1}$.


\item [(3.4)] If $j_k\leq n$ and $j_{k+1}=n+2$, then
\begin{align*}
\det(G(J_{i},[k]))&=\delta\cdot V(\boldsymbol{\alpha}_{J_{i,k+1}})\cdot\left(1+\eta\left(S_{1}^2(\boldsymbol{\alpha}_{J_{i,k+1}})-S_{2}(\boldsymbol{\alpha}_{J_{i,k+1}})\right)\right)-V(\boldsymbol{\alpha}_{J_{i,k+1}})\cdot S_{1}(\boldsymbol{\alpha}_{J_{i,k+1}})\\
&=V(\boldsymbol{\alpha}_{J_{i,k+1}})\left(\delta+\delta\eta(S_1^2-S_2)-S_1+\alpha_{j_i}\cdot(1-\eta\delta\cdot S_1)\right)
\end{align*}

and

\[
\det(G(J_{k+1},[k]))=V(\boldsymbol{\alpha}_{J_{k+1}})\cdot (1-\eta\cdot S_2).
\]
We now prove that the rank of the matrix $G_J$ is $k$ if and only if $S_1\neq \delta^{-1}\eta^{-1}$ or $S_2\neq\eta^{-1}$.  
 If $S_1\neq \eta^{-1}\delta^{-1}$, since $k-1\geq 2$, there exists $1\leq i\leq k$ such that $\delta+\delta\eta(S_1^2-S_2)-S_1+\alpha_{j_i}\cdot(1-\eta\delta\cdot S_1)\neq 0$. Thus, there exists $1\leq i\leq k$ such that $\det(G(J_i,[k]))\neq 0$. If $S_2\neq\eta^{-1}$, then $\det(G(J_{k+1},[k]))\neq 0$. Combining these two cases, we know that when $S_1\neq \eta^{-1}\delta^{-1}$ or $S_2\neq \eta^{-1}$,  the rank of the matrix $G_J$ is $k$. If the rank of the matrix $G_J$ is $k$, then there exists $1\leq i\leq k+1$ such that $\det(G(J_{i},[k]))\neq 0$. If $\det(G(J_{k+1},[k]))\neq 0$, then $S_2\neq \eta^{-1}$. If $\det(G(J_{k+1},[k]))=0$ and there exists $1\leq i\leq k$ such that $\det(G(J_{i},[k]))\neq 0$, then $S_2=\eta^{-1}$ and $\delta+\delta\eta(S_1^2-S_2)-S_1+\alpha_{j_i}\cdot(1-\eta\delta\cdot S_1)\neq 0$. Thus, $(\delta\eta S_1-1)(S_1-\alpha_{j_i})\neq 0$. So $S_1\neq\eta^{-1}\delta^{-1}$. In summary, we have proved that the rank of the matrix $G_J$ is $k$ if and only if $S_1\neq \eta^{-1}\delta^{-1}$ or $S_2\neq\eta^{-1}$.
\end{itemize}

 By Lemma~\ref{Lem:equ NMDS} and combining the above cases, we conclude that the code $\mathcal{C}$ is an NMDS code if and only if
     \begin{itemize}
      \item [(A)] for any $k-1$-subset $I\subseteq [n]$, we have $S_{1}(\boldsymbol{\alpha}_{I})\neq 0$ or $S_2(\boldsymbol{\alpha}_{I})\neq\eta^{-1}$.
          \item [(B)] for any $k$-subset $J\subseteq [n]$, we have 
      $[S_1(\boldsymbol{\alpha}_{J})\neq 0$ or $S_2(\boldsymbol{\alpha}_{J})\neq \eta^{-1}]$ and $$[S_1(\boldsymbol{\alpha}_{J})\neq \eta^{-1}\delta^{-1}\ \mbox{or}\  S_2(\boldsymbol{\alpha}_{J})\neq \eta^{-1}];$$
       \item [(C)] at least one of (i) or (ii) holds:
      \begin{itemize}
      \item[(i)] there exists a $k$-subset $L\subseteq [n]$ such that $\eta\cdot S_{2}(\boldsymbol{\alpha}_{ L})= 1$.
		\item[(ii)] there exists a $k-1$-subset ${I}\subseteq [n]$ such that $1+\eta(S_{1}^2(\boldsymbol{\alpha}_{ I})-S_{2}(\boldsymbol \alpha_{ I}))= 0$ or
        \[
        \delta\left(1+\eta(S_{1}^2(\boldsymbol \alpha_{ I})- S_{2}(\boldsymbol \alpha_{ I}))\right)-S_{1}(\boldsymbol \alpha_{ I})= 0.
        \]
        \end{itemize}
      \end{itemize}

\end{proof}

Finally, we give examples of codes $\operatorname{TRL}_{k,n+2}(\boldsymbol{\alpha},\eta,\delta,\ell)$ that are NMDS codes for $k-2\leq\ell\leq k-1$.

\begin{example}
   
   \begin{description}
   \item [\normalfont(1)] Let $n=11,k=5,\eta=\delta=1$ and $\boldsymbol{\alpha}=\{1,2,\cdots,11\}\subseteq\mathbb{F}_{37}$. For $\ell=4$, Magma verifies that there does not exist a $5$-subset $J\subseteq [11]$ such that $S_1(\boldsymbol{\alpha}_J)=-1,S_{2}(\boldsymbol{\alpha}_{J})=-1$ and there exists a $5$-subset $\mathcal{I}=\{3,6,8,9,10\}\subseteq [11]$ such that $S_{1}(\boldsymbol{\alpha}_{I})=-1$. Therefore, by Theorem~\ref{Thm: NMDS,ell=k-1}, the code $\operatorname{TRL}_{5,13}(\boldsymbol{\alpha},\eta,\delta,4)$ is an NMDS code. Additionally, direct verification by Magma shows that the minimum distances of this code and its dual code are $n-k+2=8$ and $k=5$, respectively. Thus, the codes $\operatorname{TRL}_{5,13}(\boldsymbol{\alpha},\eta,\delta,4)$ are  NMDS codes with parameters $[13,5,8]$ over $\mathbb{F}_{37}$.
   \item [\normalfont(2)] Let $n=11,k=5,\eta=3,\delta=1$ and $\boldsymbol{\alpha}=\{1,2,\cdots,11\}\subseteq\mathbb{F}_{37}$. For $\ell=3$, on the one hand, Magma verifies that for any $4$-subset $I\subseteq [n]$, we have $S_{1}(\boldsymbol{\alpha}_{I})\neq 0$ or $S_2(\boldsymbol{\alpha}_{I})\neq\eta^{-1}$ and 
   for any $5$-subset $J\subseteq [11]$, we have
      $[S_1(\boldsymbol{\alpha}_{J})\neq 0$ or $S_2(\boldsymbol{\alpha}_{J})\neq 25]$ and $[S_1(\boldsymbol{\alpha}_{J})\neq 25$ or $S_2(\boldsymbol{\alpha}_{J})\neq 25]$. On the other hand, 
      there exists a $5$-subset $\mathcal{L}=\{3,4,5,7,10\}\subseteq [11]$ such that $ S_{2}(\boldsymbol{\alpha}_{\mathcal L})= 25$.
		 Therefore, by Theorem~\ref{Thm: NMDS,ell=k-2}, the code $\operatorname{TRL}_{5,13}(\boldsymbol{\alpha},\eta,\delta,3)$ is an NMDS code. Additionally, direct verification by Magma shows that the minimum distances of this code and its dual code are $n-k+2=8$ and $k=5$, respectively. Thus, the codes $\operatorname{TRL}_{5,13}(\boldsymbol{\alpha},\eta,\delta,3)$ are  NMDS codes with parameters $[13,5,8]$ over $\mathbb{F}_{37}$.
   \end{description}
\end{example}

\section{Conclusion}
In this paper, we introduce the concept of $(\mathcal{L},\mathcal{P})$-$\operatorname{TRL}$ codes, investigate the minimum distance of a class of twisted Roth–Lempel codes $\operatorname{TRL}_{k,n+2}(\boldsymbol{\alpha},\eta,\delta,\ell)$ and use these results to obtain additional constructions of non-RS MDS codes. Specifically, our main contributions are summarized as follows.
\begin{itemize}
    \item   We determine the necessary and sufficient conditions for the code $\operatorname{TRL}_{k,n+2}(\boldsymbol{\alpha},\eta,\delta,\ell)$ to have minimum distance $n-k$ for $0\leq\ell\leq k-3$, and determine the necessary and sufficient conditions for the code to have minimum distance $n-k+1$ for $0\leq\ell\leq k-1$. 
    \item  We determine the necessary and sufficient conditions for the code $\operatorname{TRL}_{k,n+2}(\boldsymbol{\alpha},\eta,\delta,\ell)$ to be an MDS code for $0\leq\ell\leq k-1$.
    \item We show that the dimension of the Schur square of the code $\operatorname{TRL}_{k,n+2}(\boldsymbol{\alpha},\eta,\delta,\ell)$ is at least $2k+1$, and thus the code $\operatorname{TRL}_{k,n+2}(\boldsymbol{\alpha},\eta,\delta,\ell)$ is a non-RS code inequivalent to the corresponding RL code.
    
    \item Based on the TRL framework, we construct a family of non-RS MDS codes within the TRL class.
    \item  We determine the necessary and sufficient conditions for the code $\operatorname{TRL}_{k,n+2}(\boldsymbol{\alpha},\eta,\delta,\ell)$ to be an NMDS code for $k-2\leq\ell\leq k-1$.

\end{itemize}




\bibliographystyle{plain}
\bibliography{TRL}
\end{document}